\documentclass[aps,pra, twocolumn]{revtex4-1}
\usepackage{amsthm}

\usepackage[utf8]{inputenc}
\usepackage[sc,osf]{mathpazo}
\usepackage{amsmath}
\usepackage[T1]{fontenc}
\usepackage{latexsym}
\usepackage{amssymb}
\usepackage[colorlinks=true,citecolor=blue,urlcolor=blue]{hyperref}
\usepackage{color}
\usepackage{graphics,epstopdf}
\usepackage{soul}
\usepackage[demo]{graphicx}
\usepackage{capt-of}
\usepackage{lipsum}
\usepackage{adjustbox}
\usepackage[normalem]{ulem}
\usepackage[table,xcdraw]{xcolor}

\newcommand{\be}{\begin{equation}}
\newcommand{\ee}{\end{equation}}
\newcommand{\ba}{\begin{eqnarray}}
\newcommand{\ea}{\end{eqnarray}}

\newtheorem{theorem}{Theorem}

\newtheorem{proposition}{Proposition}

\begin{document}
	
	\title{Emergence of Monogamy under Static and Dynamic Scenarios}
	\author{Rivu Gupta$^{1}$, Saptarshi Roy$^{1}$, Shiladitya Mal$^{1, 2}$, Aditi Sen (De)$^{1}$}
	
	\affiliation{$^1$ Harish-Chandra Research Institute and HBNI, Chhatnag Road, Jhunsi, Allahabad - 211019, India\\
	$^2$ Department of Physics and Center for Quantum Frontiers of Research and Technology (QFort),
National Cheng Kung University, Tainan 701, Taiwan
	}
	
	\date{\today}
	
	\begin{abstract}

Characterizing multipartite quantum correlations beyond two parties is of utmost importance for building cutting edge quantum technologies, although the comprehensive picture is still missing. Here we investigate quantum correlations (QCs) present in a multipartite system by exploring connections between \textcolor{black}{three different frameworks, namely,} the monogamy score (MS), localizable quantum correlations (LQC), and genuine multipartite entanglement (GME) content of the state . We  find that the frequency distribution of GME for Dicke states with higher excitations resembles that of random states. We show that there is a critical value of GME beyond which all states become monogamous and it is investigated by considering different powers of MS which provide various layers of monogamy relations. Interestingly, such a relation between LQC and MS as well as GME does not hold. States having a very low GME (low monogamy score, both positive and negative) can localize a high amount of QCs in two parties. We also provide an upper bound to the sum of bipartite QC measures including LQC for random states and establish a gap between the actual upper bound and the algebraic maximum.

	\end{abstract}

	\maketitle

	\section{Introduction}

	Correlations play a fundamental role in providing insight into the laws describing nature at various scales. The features possessed by these correlations depend on the theory under which they have been  analyzed -- some characteristics are common to all the theories while others are exclusive to a particular one.
	Correlations in the quantum domain, commonly referred to as	quantum correlations (QCs), possess many such unique characteristics that are qualitatively different from classical correlations (CCs). From entanglement to nonlocality \cite{HoroRMP, BellRMP}, these special properties constitute and, in turn, help to understand the intricacies of quantum mechanics. Importantly, these specialties of QCs are responsible for fueling tasks like quantum teleportation \cite{teleoriginal}, quantum dense coding \cite{bennettwiesner}, genuine randomness certification \cite{randomness}, quantum computation \cite{onewayQC}, etc. which are impossible via the sole use of CCs.
	
	
Rapid developments in realizing quantum technologies demand complete characterization of multisite entangled states. 	
	\textcolor{black}{
  Since QCs shared between multiple parties might be of different types, it is often difficult to determine a concrete way to assess quantumness in a multipartite state.
  Depending on the available resource, a number of quantum information schemes using shared multipartite quantum correlated states have been designed. For example, 
  For instance, in order to obtain a quantum advantage in both classical and quantum information transmission, it is necessary for multiple senders and receivers to share states that have QCs in the sender-receiver bipartition \cite{bennettwiesner, teleoriginal}, whereas in measurement-based quantum computation, genuine multiparty entangled states become a resource \cite{onewayQC}. Therefore, quantifying QCs from different perspectives and establishing links between them can be crucial for developing quantum technologies. In this work, we focus on three aspects of multipartite QC measures, belonging to three distinct paradigms - $(1)$ genuine multiparty entanglement content, which can be quantified via the geometric structure of multipartite quantum states \cite{Wei_PRA_2003}; $(2)$ monogamy of quantum correlations which constraints the ability of multiple quantum parties to share correlations \cite{CKW, monoreview}; $(3)$ measurement based correlations in which QC is concentrated in fewer parties by performing suitable measurements on the other remaining parties \cite{entass, LC1, LC2}. 
  To carry out the investigation,  we introduce two quantities, the critical GGM, and the critical exponent $\alpha$, beyond which all states satisfy the monogamy inequality,  given by $\delta_{\mathcal{Q}^\alpha} = \mathcal{Q}^{\alpha} (\rho_{1:\mbox{rest}}) - \sum_{i=2}^N \mathcal{Q}^\alpha (\rho_{1:i})$, where $\mathcal{Q}$ is the quantum correlation measure under consideration for a given \(N\)-party state $\rho_{1\ldots N}$, \(\text{rest}\)  in the subscript defines the rest of the parties except the first party and \(\alpha\) is referred to as critical exponent. \\
  We address two important questions to establish a connection between all the correlation measures - (i) is there a threshold on the multiparty entanglement content of Haar-uniformly simulated states above which they all satisfy some monogamy relation for a fixed critical exponent?   It will be interesting to determine if the correlation measure and the exponent used to establish the monogamy relation affect the results. (ii) What is the minimum amount of genuine multiparty entanglement required to localize a high value of bipartite QC? Answering these questions initially requires independent analysis of these measures for randomly generated states which we will study here (see \cite{unimonoSappy} and localizable entanglement for $\alpha = 1$ \cite{Banerjee_PRA_2020, Banerjee_arXiv_2020}) and we will study in this paper. Additionally, we will establish a connection between these seemingly unrelated measures.
   Apart from their fundamental significance, genuine multipartite QC measure, monogamy relations and localizable correlations also possess some utilitarian applications like distinguishing classes of quantum states \cite{WHSV}, in quantum cryptography \cite{dist2, EkertCrypto, Pawlowski} and characterizing phases in many-body systems \cite{monoreview, Monoexpt, LC1, LC2, LC3, LC4, Amitda, LCASDUS, spinLC} (see also  \cite{Koashi, Illuminati2, Aditicap, PianiBruss, Fanchini1, largemono, BaiXuWang, Bartosz, VerstraeteOsborne, salini, HimadriAsutosh, MonoGHZW, SappyTamo, MonoFaithful, GourGuo, sudipto, MonoQubit, LuoLi, TightMono, MonoTri, MonoNeg, MonoNeg2, lightcone, forbidden, bellcorr, tricorrbibell, monocalc, monoreview, tehralMono, unimonoSappy, Ratul}). }\\
	\begin{figure}[h]
		\centering
		\includegraphics[width=\linewidth]{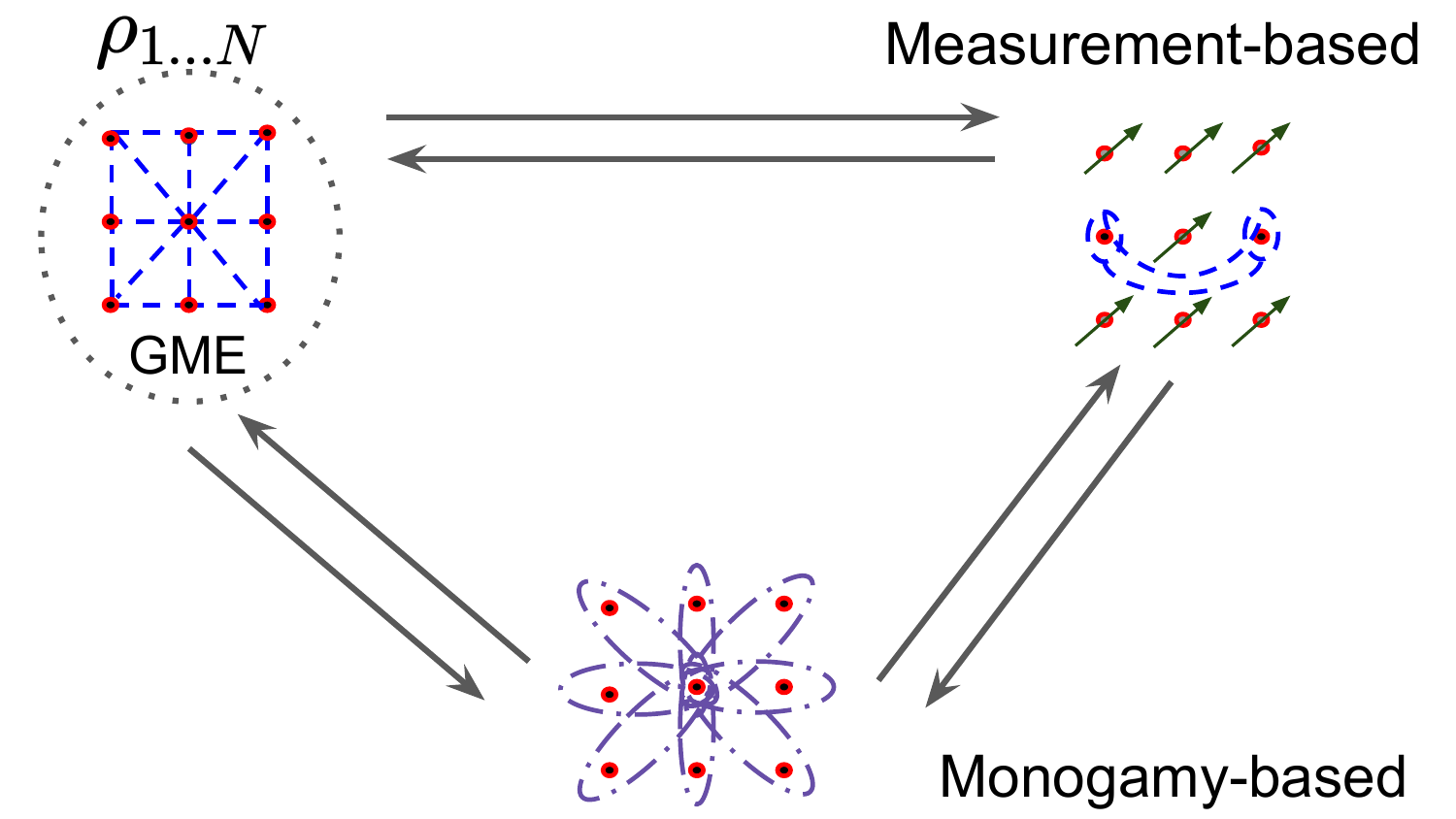}
		\caption{Schematic representation of interplay between different  multipartite quantum correlations, thereby providing classifications  among multipartite quantum correlation measures -- genuine multipartite entanglement, monogamy- and measurement-based quantum correlations are related. Analysis shows that features of multipartite QCs are more prominently present in monogamy-based measures compared to the measurement-based ones, considered in this paper. }
		\label{schematic}
	\end{figure}

 We  report that the critical GGM required to satisfy monogamy inequality decreases with the number of parties for all quantum correlation measures, thereby providing sufficient criteria, independent of QC measures. On the other hand, for states with less number of qubits and a fixed amount of GGM, the critical exponent can assume very high values but as the number of qubits grows, it saturates to its lower limits, thereby showing the increase of quantumness in randomly simulated states with a number of parties.\\	
	We prove that for a fixed amount of genuine multipartite entanglement (GME) content of an arbitrary three-qubit state and the generalized Greenberger-Horne-Zeilinger (gGHZ) state \cite{GHZ}, the monogamy score of entanglement for the former is always lower than that of the gGHZ state. Such an upper bound obtained from the gGHZ state does not hold for states having more than three qubits.  We also show that  Dicke states with higher excitations behave more like the random states while the Dicke states with single or low excitations are not. A usual way of analyzing monogamy is to look at the sum  of all possible bipartite correlations with respect to a particular party of a multipartite state. The monogamous feature is reflected as an upper bound to this sum which turns out to be much smaller than the algebraic maximum of the sum. We provide an estimate of the sum for different bipartite quantum correlation measures of random multipartite states as well as Dicke states, thereby revealing the gap between the algebraic maximum and the actual value which leads to the violation of monogamy inequality.  	
%

 In the case of measurement-based QC measures, we show that even if the original state has low GME as well as low (both positive and negative) monogamy score, substantial quantum correlation can be localized using projective measurements, which becomes more pronounced in  the case of  states having more number of qubits.  We support such observation both qualitatively and  quantitatively by considering the minimum localizable QC  produced from states with a fixed amount of  GME or monogamy score.  
	A slightly contrasting  behavior is observed for randomly generated Dicke states with a single excitation for which states possessing high GME (monogamy score) can always produce a moderate amount of localizable entanglement  although low GME can also achieve high localizable entanglement. This is due to the fact that the sample space of Dicke states having high GGM is low in number with the increase of the number of parties. 
	 Both the results illustrate that the monogamy score can capture the features of multipartite QCs more prominently compared to the measurement-based QCs, thereby showing the interplay between measurement - and monogamy-based measures with GME. We also report that, unlike monogamy scores,  the sum of the localizable QCs of  multipartite random states can reach close to their algebraic maximum, especially for states with a low number of qubits. 
%

 The paper is organized in the following way.  Sec. \ref{sec_monogen_ent} establishes a relation between monogamy of QCs and genuine multipartite entanglement for  random multipartite states by varying parties from three to six, Dicke states  with different excitations  \cite{Dicke} and three-qubit W-class states  \cite{DurVidalCirac}.  We characterize the set of states which are non-monogamous with respect to certain bipartite QC measures, in terms of genuine multipartite entanglement content  in SubSec. \ref{subsec_critGGM}. In Sec. \ref{sec_LE}, we finally relate the three quantities, monogamy score, localizable entanglement, and GGM as well as report an upper bound on the distribution of localizable entanglement. The summary of results and their implications are presented in Sec. \ref{sec_conclu}.

\section{Monogamy vs. Genuine multipartite entanglement}
\label{sec_monogen_ent}

Before providing  the relation, let us first present the prerequisites to carry out the investigation. We first give definitions of monogamy score of an arbitrary QC measure, classes of multiqubit states under study and genuine multipartite entanglement measure. 

\emph{Monogamy of QC. } The restrictions on the distribution of bipartite quantum correlations,  $\mathcal{Q}$, in a multiparty state,  $\rho_{1\ldots N}$,  is referred to as the monogamy of QCs. Quantitatively, it constrains the sum of all bipartite QCs of a quantum state with a given nodal party, say, $1$,  i.e., it provides  an upper bound, $\mathcal{Q}(\rho_{1:\text{rest}})$,  on   $\sum_{i=2}^N \mathcal{Q} (\rho_{1:i})$  where  without loss of  generality, we assume the nodal party to be the first party. 
 Hence, a state is said to be monogamous with respect to \(\mathcal{Q}\) if it  satisfies $\mathcal{Q}(\rho_{1:\text{rest}}) \geq \sum_{i=2}^N \mathcal{Q}(\rho_{1:i})$. This is evaluated via the monogamy score, which for any power, $\alpha$, of a given $\mathcal{Q}$, is defined as \cite{vanmon}
\begin{equation}
	\delta_{\mathcal{Q}^\alpha} = \mathcal{Q}^{\alpha}_{1:\mbox{rest}} - \sum_{i=2}^N \mathcal{Q}_{1:i}^\alpha,
	\label{eq:monogamy}
	\end{equation}
	where \(\mathcal{Q}^\alpha_{1:\mbox{rest}} \equiv \mathcal{Q}^\alpha(\rho_{1:\text{rest}})\) and  $\mathcal{Q}^\alpha_{1:i} \equiv  \mathcal{Q}^\alpha (\rho_{1:i})$. 
In this work, the  QC measures are considered to be  negativity ($\mathcal{N}$), concurrence ($\mathcal{C}$) and quantum discord ($\mathcal{D}$). 

\emph{Simulation of quantum states.} 
An N-qubit random pure state chosen Haar uniformly reads as \cite{Zyczkowskibook}
	\begin{equation}
	|\psi_R \rangle = \sum_{i=1}^{2^N} \xi_i|i_1 i_2...i_N \rangle
	\label{eq:rand_state}
	\end{equation}
	where $\xi_j = a_j + i\, b_j$ with $a_j \; \text{and} \; b_j \in \mathbb{R}$ being sampled from a Gaussian distribution of mean $0$ and unit standard deviation ($\mathbb{G}(0,1)$) and $\{|i_k\rangle\} $s constituting the computational basis. For $N = 3$, the state space splits into two  inequivalent classes of states under stochastic local operations and classical communication, the GHZ- and the W-class states \cite{DurVidalCirac}. The GHZ class states take  the same form as in Eq. \eqref{eq:rand_state}, while the W-class states, constituting a set of measure zero are given by
	\begin{equation}
	|\psi_W \rangle = a |000\rangle + b |001\rangle + c |010\rangle + d |100\rangle,
	\label{eq:Wstate}
	\end{equation}
	where \(a, b, c, d\) are complex numbers whose real parts are  taken from $\mathbb{G}(0,1)$. 
	For states with higher number of qubits, i.e., for  $(N\geq 3)$, we consider another class of states, the Dicke states \cite{Dicke}, which reduces to the generalized W state (obtained from Eq. (\ref{eq:Wstate})  by putting \(a=0\)) for three-qubit case.  A Dicke state of $N$ qubits having $r$ excitations is defined as
	\begin{equation}
	|\psi_D^r\rangle = \sum c_\mathcal{P} \mathcal{P}(|0\rangle^{\otimes (n-r)}\otimes |1\rangle^{\otimes r}),
	\label{eq:Dicke}
	\end{equation}
	where  \(\mathcal{P}\) denotes the permutation of all states with \(n-r\) excitations,  $|1\rangle$ and \(r\) ground states,  $|0\rangle$. The coefficients $c_\mathcal{P} = c_{1_\mathcal{P}} + i c_{2_\mathcal{P}}$ are again chosen from $\mathbb{G}(0,1)$ during their simulation, so that random Haar uniformly  chosen Dicke states are numerically generated.  For four- and five-party states, the excitations are taken to be a single or two while  we have upto three excitations for six-qubit Dicke states.	
	
\emph{Genuine multipartite entanglement: generalized geometric measure.} The genuine multiparty entanglement (GME) content of these random pure states can be computed using the generalized geometric measure (GGM). It is a distance-based measure of GME and is defined as the minimum distance of a given state from the set of all non-genuinely entangled states in the state space \cite{Aditicap, ggmpure1}. For general mixed states,  carrying out the minimization is very hard \cite{ggmmixed}. However, for pure states, the Schmidt decomposition makes the optimization procedure tractable and the GGM can be expressed in terms of Schmidt coefficients in different bipartitions of the multipartite pure state,
$|\psi_N\rangle$,   as
\begin{eqnarray}
\label{Eq:GGM-final-express}
&&\mathcal{G}(|\psi_N\rangle) = \nonumber \\  &&1 - \max \big \lbrace \lambda_{\cal A:B} | {\cal A}\cup {\cal B} = \lbrace 1,2,\ldots, N \rbrace, {\cal A}\cap{\cal B} = \emptyset \big \rbrace, \nonumber \\
\end{eqnarray}
where $\lambda_{\cal A:B}$ is the maximum Schmidt coefficient in the $\cal A : B$
bipartition of $|\psi_N\rangle$, and maximization is performed over all such possible bipartitions. Before exploring the monogamy features, let us discuss some of the GGM characteristics of random states which will make it a convenient reference point when comparisons with the monogamy scores will be made.

	\begin{figure}[h]
		\centering
		\includegraphics[width=\linewidth]{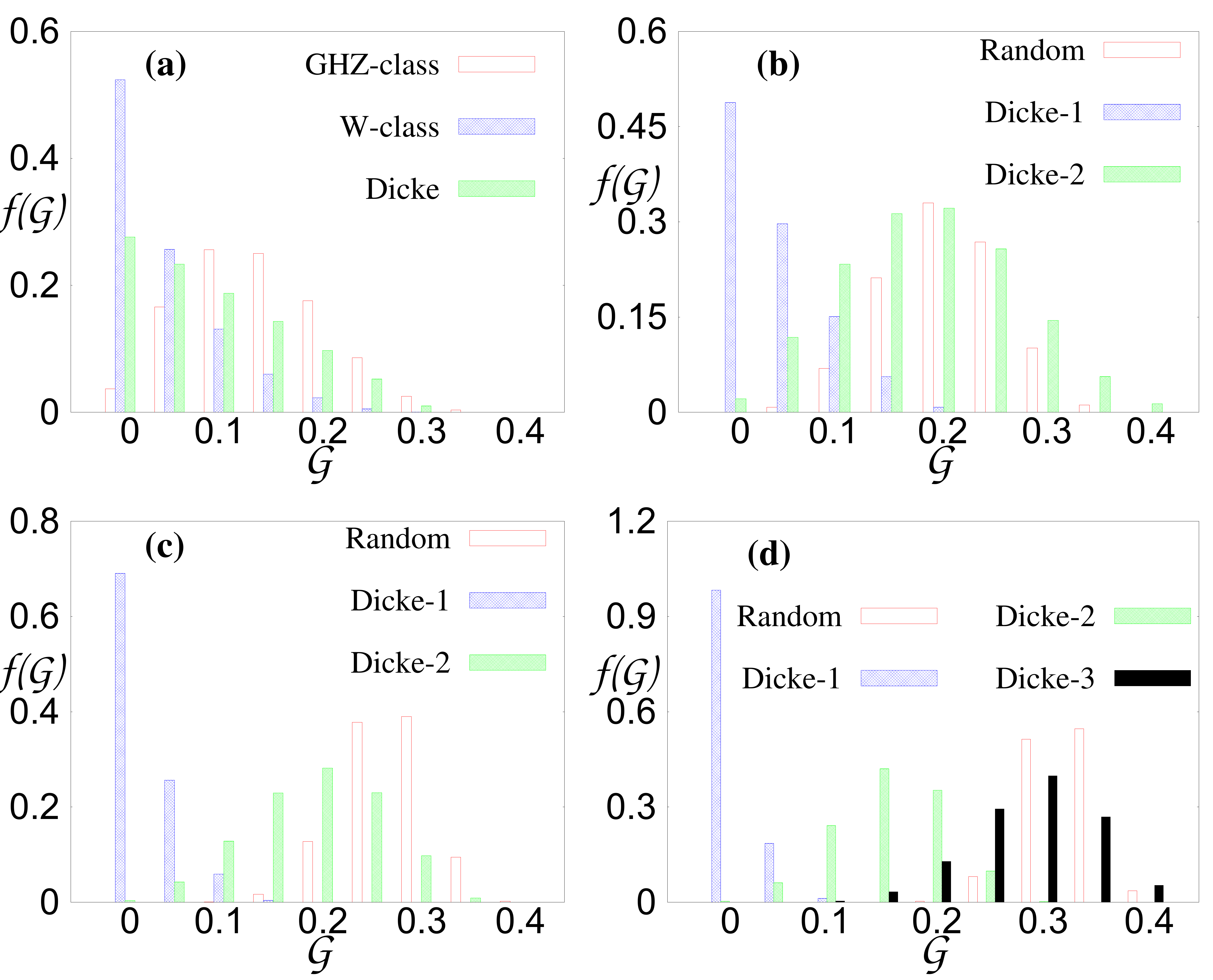}
		\caption{\textcolor{black}{(Normalized) Frequency distribution of GGM, $f(\mathcal{G})$ (ordinate) vs. the GGM, $\mathcal{G}$ (abscissa).  Haar uniformly generated random states (red), random Dicke class with single  (blue),  two (green) and three-excitations  (black) for (a) three-qubit, (b) four-qubit , (c) five-qubit  and (d) six- qubit states are generated. Number of states (all kinds) simulated is \(5\times 10^5\). The ordinate is dimensionless, and the abscissa is in ebits.}}
		\label{fig_ggm_3456}
	\end{figure}

\emph{Frequency distribution of GGM.} To calculate the frequency distribution, \(f(\mathcal{G})\), we count the number of states having GGM between, say, \(a\) and \(b\) which is then divided by the total number of states simulated. In  rest of the paper, wherever we calculate frequency distribution, we use this normalized version. 

For \emph{random pure states} of three- to six-qubits, the distribution takes a bell shape whose mean increases with \(N\) while the standard deviation (SD) decreases with the increase of  number of parties as shown in Table \ref{msdTD} and Fig. \ref{fig_ggm_3456}.  The maximum value of GGM for random states that can be simulated also increases when \(N\) increases from three to six and  it is close to its algebraic maximum, i.e., \(0.5\) for random six-qubits which can also be obtained for the \(N\)-party GHZ state \cite{Aditicap}. 

On the other hand, as one expects,  the trends in frequency distribution for  GGM are drastically different for the \emph{Dicke states} with low excitations, \textcolor{black}{where \textit{Dicke-n} represents the Dicke state with n-excitations}. In particular, if we consider Haar uniformly generated three-qubit W-states,   \(f(\mathcal{G})\)  is steadily decreasing with a peak around \(0-0.05\) which hosts half of the states. The corresponding average GGM is 0.063  with the standard deviation being 0.056. Hence, most of the states  in this class possess a low genuine multipartite entanglement, which reaches its maximum at \(0.326\). The maxima  as well as the average value of GGM for the Dicke states with a single excitation sharply decreases with the increase of the number of parties ( see  Tables \ref{msdTD} and \ref{max_ggm_N}). For example, the fraction of states residing in the  GGM bin of \(0-0.05\) increases with \(N\)-- \(50\%\)  for three-qubits, \(70\%\) for five-qubits and almost all the simulated states for six-qubits. Interestingly, with the increase of excitations in Dicke states, the distribution follows the same pattern as in the random states as we will show in the following proposition. 
	
	

	
	\begin{table}[]
		\caption{Mean and SD  of GGM, $\mathcal{G}$ for random states with different number of parties. } 
		\begin{tabular}{|l|l|l|}
			\hline
			& Mean  & SD    \\ \hline
			3 & 0.162 & 0.069 \\ \hline
			4 & 0.231 & 0.055 \\ \hline
			5 & 0.295 & 0.042 \\ \hline
			6 & 0.347 & 0.031 \\ \hline
		\end{tabular}
		\label{msdTGHZ}
	\end{table}

	\begin{table}[]
		\caption{Mean and SD of $\mathcal{G}$ for random Dicke states having different number of excitations. } 
		\begin{tabular}{|l|l|l|l|l|l|l|}
			\hline
			& \multicolumn{2}{l|}{$|\psi_D^1\rangle$} & \multicolumn{2}{l|}{$|\psi_D^2\rangle$} & \multicolumn{2}{l|}{$|\psi_D^3\rangle$} \\ \hline
			&                    &                    &                    &                    &                    &                    \\ \hline
			& Mean               & SD                 & Mean               & SD                 & Mean               & SD                 \\ \hline
			&                    &                    &                    &                    &                    &                    \\ \hline
			3 & 0.11               & 0.079              &                    &                    &                    &                    \\ \hline
			4 & 0.062              & 0.048              & 0.21               & 0.082              &                    &                    \\ \hline
			5 & 0.039              & 0.033              & 0.22               & 0.066              &                    &                    \\ \hline
			6 & 0.028              & 0.023              & 0.183              & 0.049              & 0.313              & 0.056              \\ \hline
		\end{tabular}
		\label{msdTD}
	\end{table}
	
	\	\begin{table}[]
	\caption{Actual maximum  of $\mathcal{G}$ by varying number of qubits for randomly generated and Dicke states.  } 
	\begin{tabular}{|r|r|r|r|r|}
		\hline
		\multicolumn{1}{|c|}{N} & \multicolumn{1}{c|}{Random} & \multicolumn{1}{c|}{$|\psi_D^1\rangle $} & \multicolumn{1}{c|}{$|\psi_D^2\rangle$} & \multicolumn{1}{c|}{$|\psi_D^3\rangle$} \\ \hline
		\multicolumn{1}{|l|}{}  & \multicolumn{1}{l|}{}    & \multicolumn{1}{l|}{}           & \multicolumn{1}{l|}{}           & \multicolumn{1}{l|}{}           \\ \hline
		3                       & 0.429                    & 0.33                            &                                 &                                 \\ \hline
		4                       & 0.435                    & 0.246                           & 0.45                            &                                 \\ \hline
		5                       & 0.449                    & 0.194                           & 0.397                           &                                 \\ \hline
		6                       & 0.453                    & 0.154                           & 0.325                           & 0.485                           \\ \hline
	\end{tabular}
\label{max_ggm_N}
\end{table}
	

	
	\begin{proposition}
		The average GGM of an \(N\)-qubit Dicke state with \(N/2\) excitations for even \(N\) and \(N/2 +1\) excitations for odd \(N\) is almost the same as that of the random states of \(N\)-qubits, especially in the limit $N \gg 1$. \\
	\end{proposition}

\begin{proof}

\textcolor{black}{Let us first consider the situation when \(N\) is even. The logic behind the statement remains similar for odd \(N\). For an \(N\)-qubit Dicke state comprising \(N/2\) excitations having equal coefficients, the maximum eigenvalue comes from the $2:N-2$ bipartition and is given by $N/2(N-1)$ \cite{ggmfrust}. Thus the GGM of such a state is $\mathcal{G}^{eq} = (N-2)/2(N-1)$, where the superscript "eq" indicates that the coefficients are all equal.  To obtain the GGM of an \(N\)-qubit random state, we observe from our numerical calculations that the largest eigenvalue  comes from a single-party reduced density matrix. To that end,  we try to estimate it by approximating the average value  of  the von Neuman entropy (given by \(- \mbox{tr} \rho \log_2 \rho\)) of the reduced state, using the formula \cite{Kendon}
\begin{equation}
	\langle S \rangle = \log_2(M) - \frac{M}{2K} = 1 - \frac{1}{2^{N-1}},
	\label{Zyczentropy}
	\end{equation}
	where $M = 2$ is the dimension of the density matrix of the one qubit reduced state and $MK = 2^N$ represents the total dimension of the pure state  from which the reduced system is obtained upon tracing out. Since $\langle S \rangle$ is the entropy of a single qubit state, we can find the largest eigenvalue, say, $(1 - x)$ by solving
	\begin{equation}
	-x\log_2(x) - (1-x)\log_2(1 - x) = \langle S \rangle.
	\label{entropy_largex}
	\end{equation}
	Then, the average GGM of the random state is, $\langle\mathcal{G}\rangle  = x$. \\
 In the limit of a large number of parties, i.e., $N \gg 1$, we can approximate $\mathcal{G}^{eq}$ as
 \begin{eqnarray}
     \mathcal{G}^{eq}  = \frac{1}{2} \Big(\frac{N - 2}{N - 1}\Big) = \frac{1}{2}\Big(1 - \frac{2}{N}\Big)\Big(1 - \frac{1}{N}\Big)^{-1} \approx \frac{1}{2}\Big(1 - \frac{1}{N}\Big),~~~~~~
     \label{eq:G_eq-approx}
 \end{eqnarray}
where we have ignored terms of $\mathcal{O}(\frac{1}{N^2})$ and lower. In the same limit, since $\langle S \rangle \to 1$, we have $x \to 1/2 $ and thus, we can  consider $x = \Big( \frac{1}{2} - \frac{1}{2^K}\Big)$ with $K \to \infty$. Therefore, Eq. \eqref{entropy_largex} becomes
\begin{eqnarray}
    \langle S \rangle \approx 1 - \frac{1}{2^{2K}} = 1 - \frac{1}{2^{N-1}},
    \label{eq:entropy_approx}
\end{eqnarray}
which implies $K = N-1/2$ and 	$\langle\mathcal{G}\rangle  = \frac{1}{2}\Big(1 - \frac{1}{2^{N - 1/2}}\Big)$. Here, we have used the identity $\log_2(1 \pm \epsilon) \approx \pm \epsilon$ when $\epsilon \to 0$. Therefore, the difference between the GGM of the Dicke state and the random state for large $N$ scales as
\begin{eqnarray}
    \langle\mathcal{G}\rangle - \mathcal{G}^{eq} \approx \frac{1}{2 N} - \frac{1}{2^{\frac{N + 1}{2}}} \approx \frac{1}{2N},
    \label{eq:GGM_diff}
\end{eqnarray}
which becomes vanishingly small as $N$ assumes larger and larger values, with both $\langle\mathcal{G}\rangle$ and $\mathcal{G}^{eq}$ tending towards $0.5$. Hence the proof.}
\end{proof}
	
   	
   	We will discuss the distribution of localizable QCs in subsequent sections, but before that, we shall be investigating the connection between monogamy score and GGM of a multiparty entangled state.
	

   \subsection{Relationship between monogamy score and GGM}


To establish a connection between monogamy scores   in Eq. \eqref{eq:monogamy} with respect to negativity, concurrence and quantum discord for various values of the exponent, $\alpha$ and GGM, we address the following questions:
	\begin{itemize}
		\item Is there a pattern in the distribution of non-monogamous states in terms of their GGM content? How does that depend on the exponent, \(\alpha\)? 
		
		\item Is it possible to find a critical value of GGM beyond which no non-monogamous states are present (Eq. \eqref{Gc}) and is it independent of the choice of QC measure for a fixed exponent?  An answer to this question can shed light on the properties of the non-monogamous nature of QC measures, thereby giving a sufficient condition on states satisfying monogamy relation in terms of GME. As we know, qualitatively and in an extreme situation, bipartite quantum states having maximal QCs follow the monogamy relation. A possible reason can be that the violation obtained is due to the stringent bound that we put on  \(\sum_{i=2}^{N} \mathcal{Q}_{1i}\). Hence, it will also be interesting to find the  actual upper bound on the sum for random states.
		
		\item If  the distribution of GGM with respect to monogamous and non-monogamous states is considered, depending on the set of states, how does such distribution change? 
	\end{itemize}

To examine the relational properties of randomly generated states, we define the following quantities. Firstly, we segregate the random states into bins possessing definite ranges of GGM values and compute the fraction of non-monogamous states in each bin, which, in turn, is computed as
	\begin{equation}
	f^{NM}_{\mathcal{Q}^{\alpha}} = \frac{\text{Number of non-monogamous states}}{\text{Total number of states within GGM range}},
	\label{frac_s}
	\end{equation}
	for a fixed QC measure. \textcolor{black}{The fraction of monogamous states is in turn, given by $f^{M}_{\mathcal{Q}^{\alpha}} = 1 - f^{NM}_{\mathcal{Q}^{\alpha}}$}.
	Such a quantities is useful to address the first and the last questions while we
compute the content of GGM above which  all randomly generated states turn out to be monogamous which we refer to as critical value of GGM, given by 	
	\begin{equation}
	\mathcal{G}_c = \text{maximum GGM beyond which} \; \;  \delta_Q^\alpha \geq 0,
	\label{Gc}
	\end{equation}
	to obtain the answer to the second one.
  Our aim is to find the change that occurred in the  critical GGM depending on the choice of  the QC measure, $\mathcal{Q}$, and power $\alpha$ in monogamy score. 
  
	\begin{figure}[h]
		\centering
		\includegraphics[width=\linewidth]{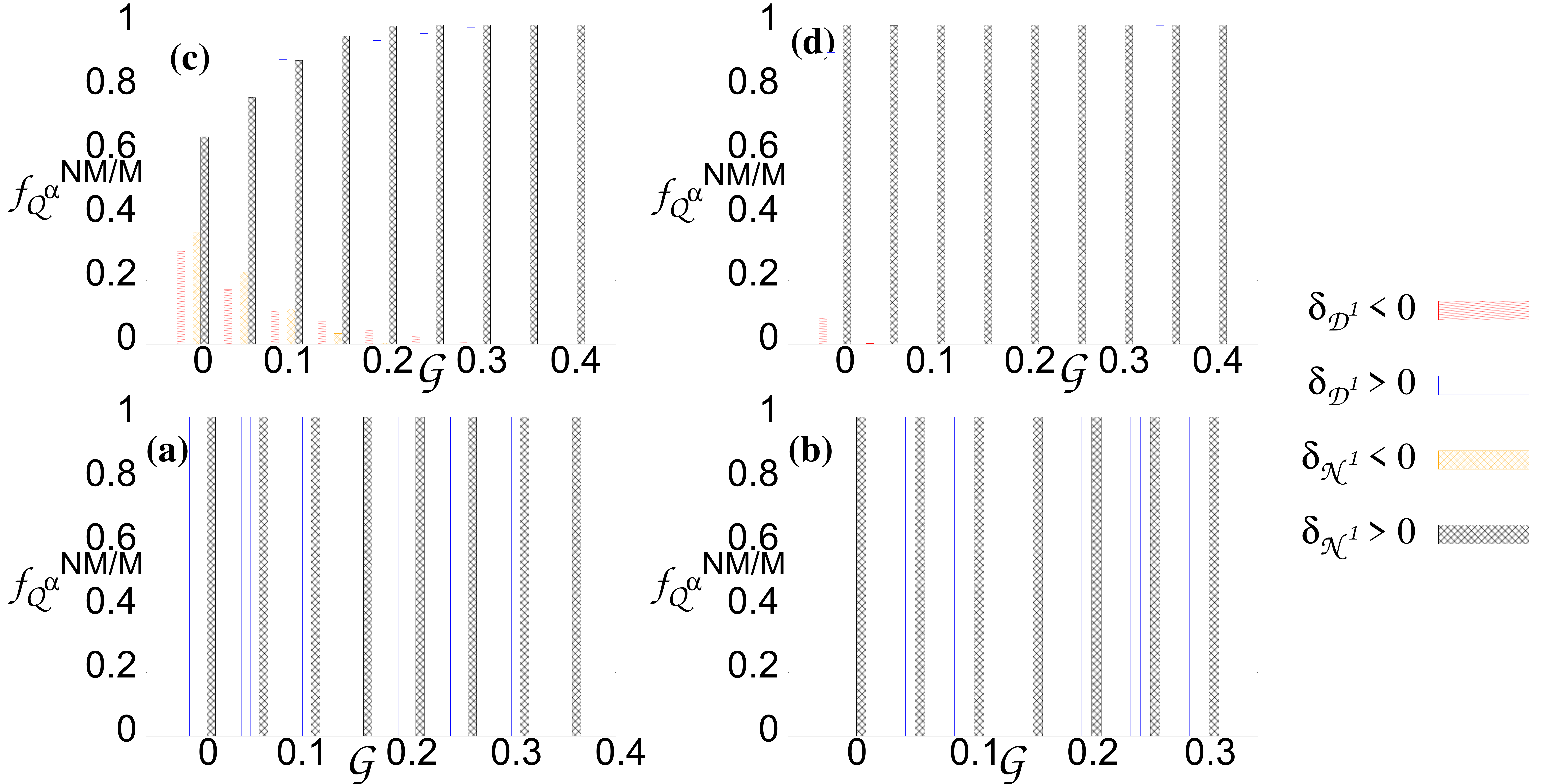}
		\caption{\textcolor{black}{Frequency distribution of both non-monogamous, $f^{NM}_{\mathcal{Q}^{\alpha}}$, and monogamous, $f^{NM}_{\mathcal{Q}^{\alpha}}$,  states (ordinate) with the GGM, $\mathcal{G}$ (abscissa). Both discord $\delta_{\mathcal{D}^1}$ and negativity $\delta_{\mathcal{N}^1}$ monogamy scores are studied with \(\alpha =1\). All other specifications are the same as in Fig. \ref{fig_ggm_3456}}. 
}
		\label{ds1g3}
	\end{figure}

	\begin{figure}[h]
		\centering
		\includegraphics[width=\linewidth]{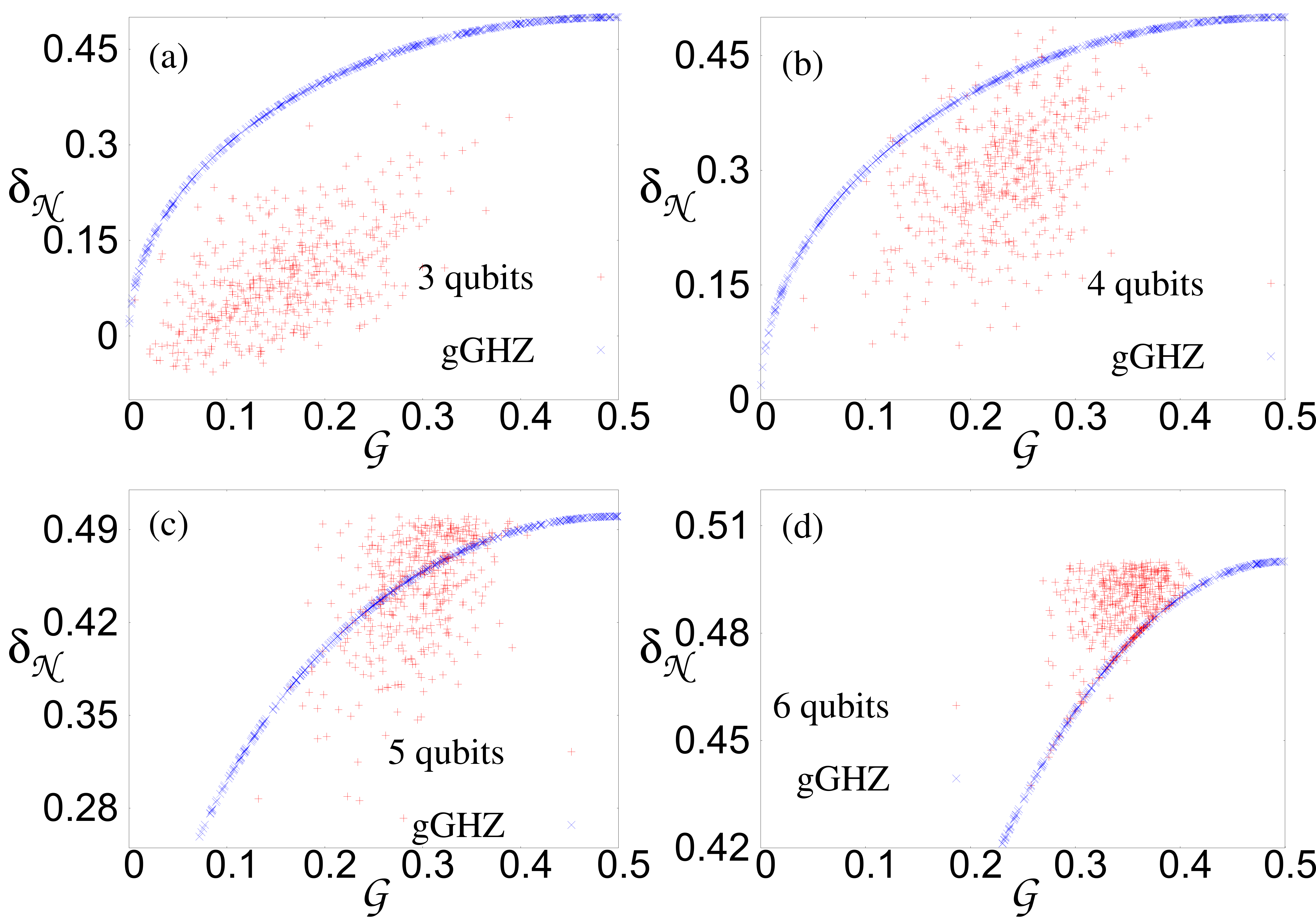}
		\caption{\textcolor{black}{$\delta_{\mathcal{N}}$ (red)  (vertical axis) against $\mathcal{G}$ (horizontal axis). Negativity monogamy scores for random three-qubit (a), four-qubit (b), five-qubit (c), and six-qubit (d)  states for a given $\mathcal{G}$ are plotted. The blue solid line represents the gGHZ state.  Both axes are in ebits.} }
		\label{neg_score_1}
	\end{figure}
	
	\begin{figure}[h]
		\centering
		\includegraphics[width=\linewidth]{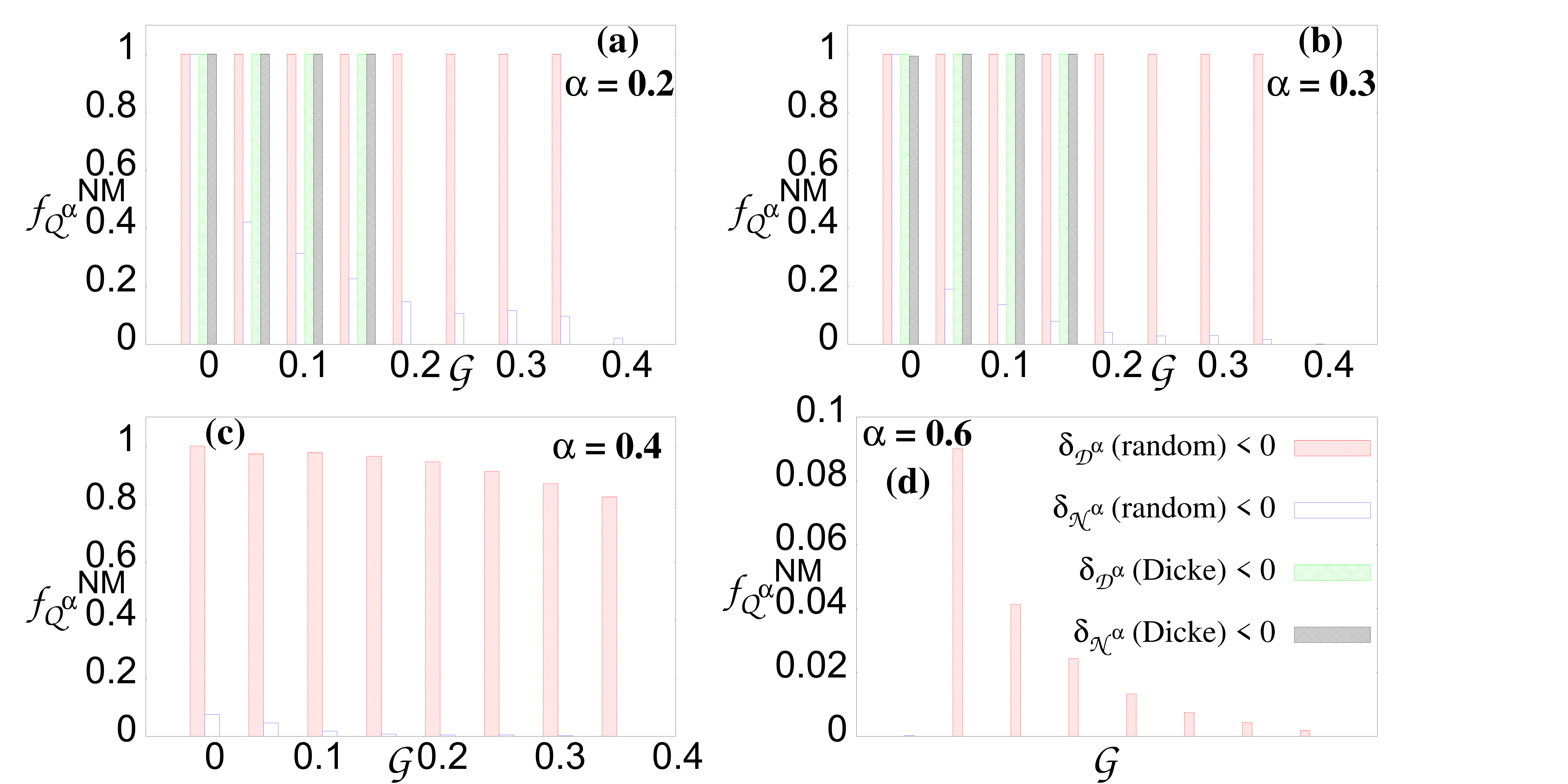}
		\caption{\(f^{NM}_{\mathcal{Q}^{\alpha}}\) (vertical axis) vs.  $\mathcal{G}$  (horizontal axis).    $\delta_{\mathcal{D}}^\alpha$ and $\delta_{\mathcal{N}}^\alpha$ are plotted for different \(\alpha\) values, for five-qubit random states (red and blue) and five-qubit Dicke states with one excitation (green and black) for (a) $\alpha = 0.2$, (b) $\alpha = 0.3$, (c) $\alpha = 0.4$, and (d) $\alpha = 0.6$.  The \(y\)-axis is dimensionless while the \(x\)-axis is in ebits.   }
		\label{disc_score_frac5}
	\end{figure}
	

	\begin{figure*}
		\centering
		\includegraphics[width=0.9\linewidth]{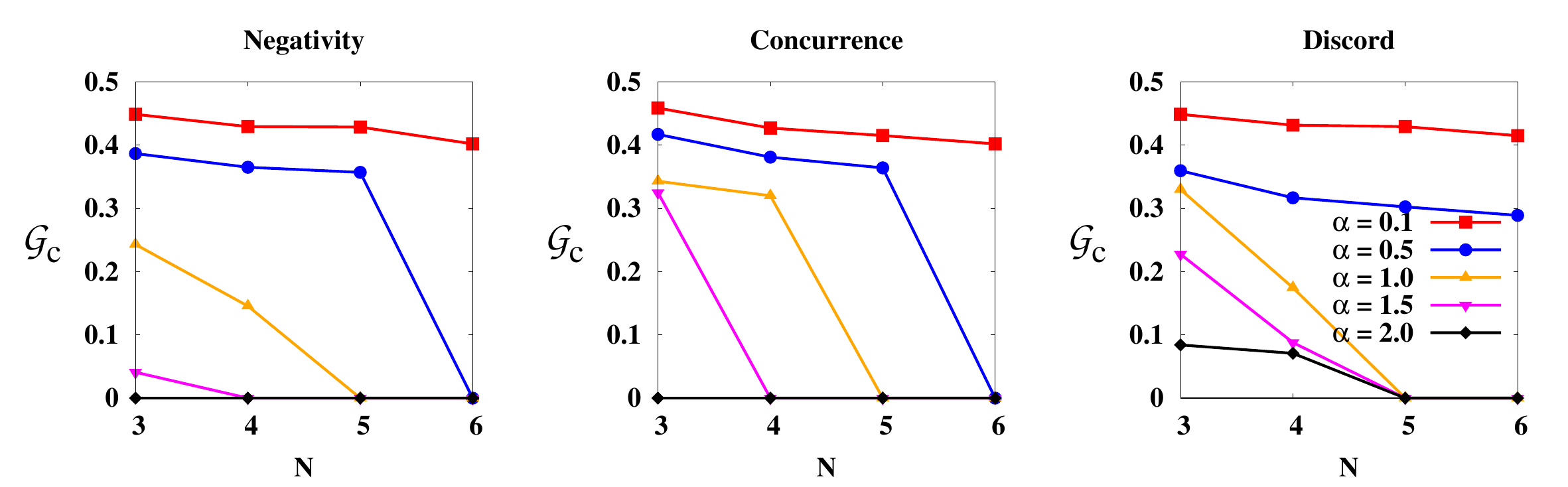}
		\caption{The critical value of GGM above which all the states are monogamous for  fixed QC measures, i.e.,  $\mathcal{G}_c$ ($y$-axis) with respect to  $N$ ( $x$-axis). 
		As a QC measure, we consider negativity (left), concurrence (middle) and discord  (right). Different exponents, \(\alpha\),  involved in monogamy are considered. The horizontal axis is dimensionless while the vertical one is in ebits. }
		\label{g_c_a}
	\end{figure*}

	
	\subsubsection{Random states} 
Let us first resolve the questions for random Haar uniformly generated states. 	As we will show, completely different picture emerges for a specific class of states. 
If we first  focus on the non-monogamous states as they vary with their respective GGM content, we observe that with increase in the number of parties, non-monogamous states cease to exist, especially for entanglement. Also, at higher values of GGM, such states decrease in number, especially in case of entanglement but for discord, $\delta_{\mathcal{D}^1}$  stays negative for a larger GGM range (see Fig \ref{ds1g3}). With an increase in the number of parties, the minimum monogamy score goes from being negative to positive and the corresponding non-monogamous states possess low amount of genuine multipartite entanglement.  This is possibly expected, since with more number of parties, the inherent quantum  correlations present in the system increase and non-monogamous states exist only at small values of multipartite entanglement as depicted in Fig. \ref{neg_score_1}. 
 Moreover, we find that  the gGHZ state provides an upper bound for three-qubit pure states which we will  prove analytically both for negativity- and concurrence-monogamy score. Interestingly, as shown in  Fig. \ref{neg_score_1}, the upper bound does not hold with an increase in the number of parties.

	\begin{theorem}
	For  random three-qubit pure states, $|\psi \rangle$, which have the same GGM as the generalized GHZ state $|\psi^{GG}_3\rangle$, the entanglement monogamy score is bounded above by that of the gGHZ state.
\end{theorem}

\begin{proof}

The reduced density matrices of the gGHZ state, $|\psi^{GG}_3\rangle = \beta |000 \rangle + \gamma |111 \rangle$, are separable and hence the monogamy score for negativity reduces to 
\begin{eqnarray}
\delta_\mathcal{N}^{GG} = \mathcal{N}^{GG}_{1:23} = \sqrt{\beta^2 (1 - \beta^2)},  
\label{delta_3_GG}
\end{eqnarray}
assuming \(\beta^2 \geq 1/2 \geq \gamma^2\), 
while  its GGM is always given by $\mathcal{G}(|\psi_{GG}^3\rangle) = 1 - \beta^2$, since it is symmetric with respect to the permutation of parties. On the other hand, suppose the tripartite state has Schmidt coefficient, \(\lambda_1^2 \geq 1/2 \) in \(1:rest\)-bipartition and GGM comes from that bipartition. If the GGM of gGHZ and arbitrary tripartite state coincide,  we have $\lambda_1^2 = \beta^2$. Moreover, 
\begin{eqnarray}
&& \delta_\mathcal{N} \leq  \mathcal{N}_{1:23} = \sqrt{\lambda_1^2 (1 - \lambda_1^2)} = \delta_\mathcal{N}^{GG}
 \label{delta_3_R} 
\end{eqnarray}
and hence the proof. In a similar fashion, one can get the proof for concurrence as for pure states,  negativity and concurrence are different by a factor of $2$. 

%
%

 Let us assume that the largest eigenvalue contributing to the GGM of the random state comes from a  party, other than the nodal party, i.e.,  $\mathcal{G}(|\psi_{R}^3\rangle) = 1 - \lambda_2^2$ with \(\lambda_1 <  \lambda_2\) and \( \lambda_2^2 \geq 1/2\). Again, we find  $\lambda_2^2 = \beta^2$ by equating GGM for the gGHZ and arbitrary state. Thus, from Eq. \eqref{delta_3_R}, we have $\delta_\mathcal{N} \leq  \sqrt{\lambda_1^2 (1 - \lambda_1^2)} \leq \sqrt{\lambda_2^2 (1 - \lambda_2^2)}  = \sqrt{\beta^2 (1 - \beta^2)} = \delta_\mathcal{N}^{GG}$ and the second inequality is due to the fact that \( \lambda_2^2 \geq 1/2\). 
	

\end{proof}

The relation between monogamy score and GGM changes drastically when the power involved in monogamy is taken less than unity \cite{unimonoSappy}. Specifically, 
when $\alpha \leq 0.5$, all states violate monogamy relation, i.e. $\delta_{Q^{\alpha \leq 0.5}} < 0$ and states having high genuine multipartite entanglement can also violate the monogamy relation as depicted in Fig. \ref{disc_score_frac5}. For $0.5 < \alpha <1.0 $,  fraction of such states decreases and again similar pattern as described before for \(\alpha =1\) emerges. 
 For $\alpha > 1.0$, almost all states are monogamous, especially for states with five or more qubits, irrespective of QC measures 	(see Fig. \ref{disc_score_frac5}) \cite{largemono}. \\
	
	\subsubsection{W class} 
	Among the states from the three-qubit W-class, the range of multipartite entanglement for which non-monogamous states exist is larger for a given QC measure, than the random states. The fraction of such states is also larger for a particular GGM interval. Thus, the critical GGM is also higher in this case compared to random states.
	
	Considering negativity and concurrence, we see that when $\alpha < 1$, a significant percentage of states remains non-monogamous while with $\alpha \geq 1$, the number of such states decreases with $\mathcal{G}$ but non-monogamous states exist for substantially high values of the exponent, upto $1.9$. In case of discord, however, states violating the monogamy inequality exist for all values of  exponent upto \(\alpha =3\), although the number is decreasing with GGM for $\alpha > 1.0$, provided the measurements are done on the nodal party, i.e., the first party in our case.

	\subsubsection{Dicke states} As the number of excitations and parties increase, the situation is similar to the random states as already argued for GGM. 
	The non-monogamous states fall in fraction more and more sharply and with an increase in the number of parties, the GGM range for the existence of such states also decreases. For five- and six-qubits, all states become monogamous for two or more excitations when $\alpha \geq 1$. It indicates that multipartite quantum correlations get enhanced with an increase in excitations and they behave in a similar fashion to  random states. Similarly, Dicke states having  high excitations and multipartite entanglement content can violate monogamy score with low \(\alpha\) which does not remain true when \(\alpha\) is increased. On the other hand, Dicke states with a single or low excitations show a large fraction of states to be non-monogamous even for a moderate \(\alpha\) values. (see Fig. \ref{disc_score_frac5})



%

	\subsection{Criticalities in GGM and Monogamy power }
	\label{subsec_critGGM}
	
	An interesting feature in the relationship between GGM and monogamy score is the existence of a critical value of GGM, $\mathcal{G}_c$ as defined in Eq. \eqref{Gc}. It means that if a random state possess a GGM value above $\mathcal{G}_c$, it is guaranteed to be monogamous. We track the changes in the values of $\mathcal{G}_c$ with the number of  parties, $N$ and the monogamy power $\alpha$.
	
	When the monogamy power is set to unity, i.e., $\alpha = 1$, we find that all QC measures show similar features, where the $\mathcal{G}_c$ decreases with $N$, hitting zero for $N = 5$, see Fig. \ref{g_c_a}. For  $\alpha$  values different from unity, we get varying responses of $\mathcal{G}_c$, as seen in Fig. \ref{g_c_a}. 
	
	To associate \(\mathcal{G}_c\) with \(\alpha\) in monogamy score, for a fixed multipartite entanglement content of a state, we find a critical exponent 
	 beyond which, the monogamy score is always non-negative. We denote it by $\alpha_C$. To enunciate its variation with GGM and its dependence on the number of parties, we consider negativity and discord as the correlation measures. Based on the observations from Fig. \ref{alpha_C}, we note the following points:
	\begin{enumerate}
		\item States which require a high value of $\alpha$ to satisfy the monogamy inequality are present for low number of parties and the number of such states decreases significantly for $N \geq 5$.
		
		\item $\alpha_C \geq 1$ exist only for very low $\mathcal{G}$. This is because, states possessing significant  genuine multipartite correlations are monogamous over a large range of the exponent.
		
		\item Near the tail of the GGM spectrum, where states are strongly quantum correlated, $\alpha_C$ is low, even for low number of qubits which is nicely depicted in Fig. \ref{alpha_C} for three-qubits.
		
		\item Non-monogamous states are mostly observed for $\alpha \leq 1$ for all multi-qubit regimes, independent of the choice of QC measures.
		
	\end{enumerate}

	\begin{figure}[h]
		\centering
		\includegraphics[width=\linewidth]{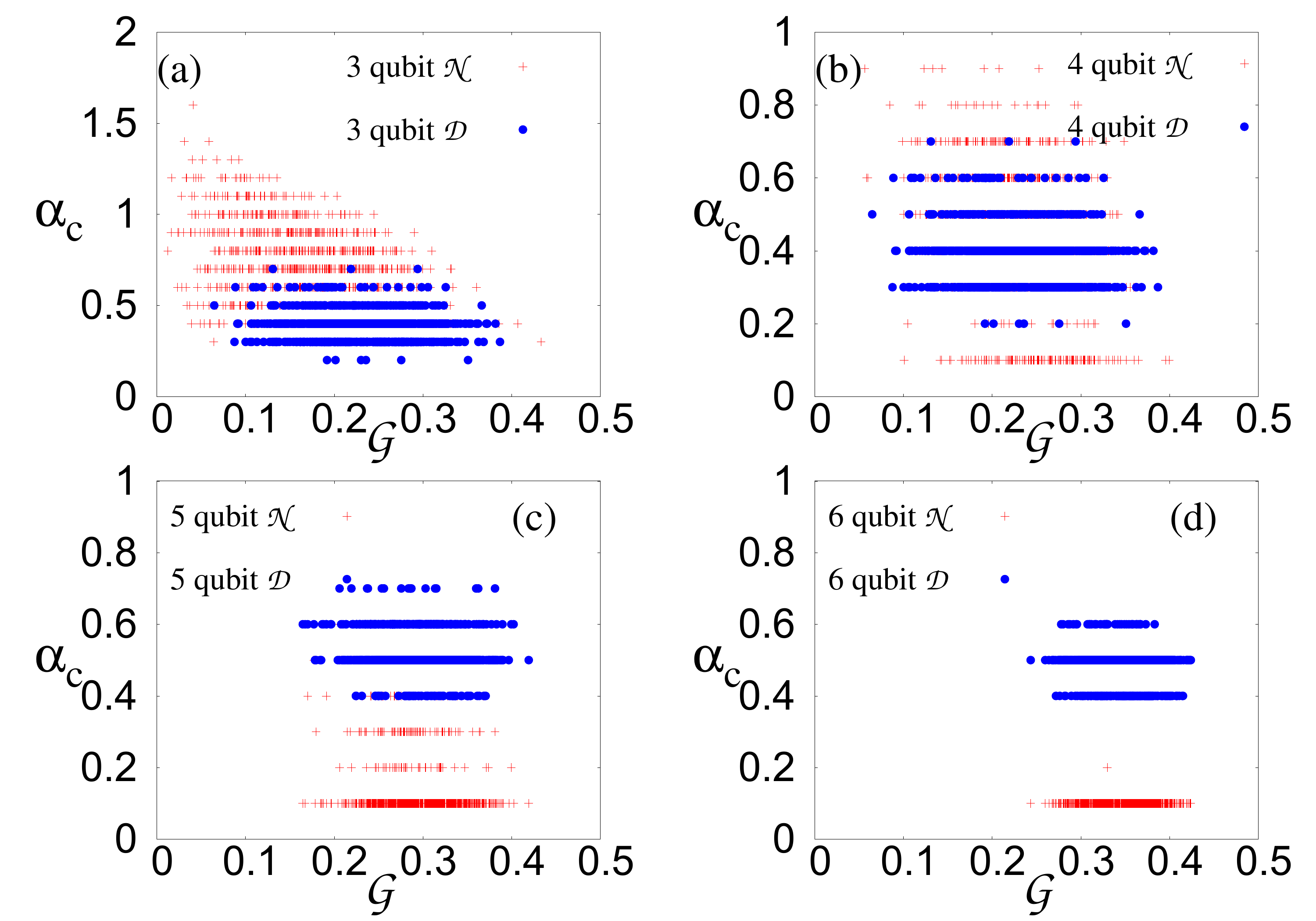}
		\caption{Critical exponent, $\alpha_C$  (see text for definition) along ordinate against $\mathcal{G}$ along abscissa both for discord  (solid circles) and negativity (pluses) monogamy scores. All other specifications are same as in Fig. \ref{fig_ggm_3456}. 
		}
		\label{alpha_C}
	\end{figure}
	
\begin{figure*}
\centering
\includegraphics[width=0.85\linewidth]{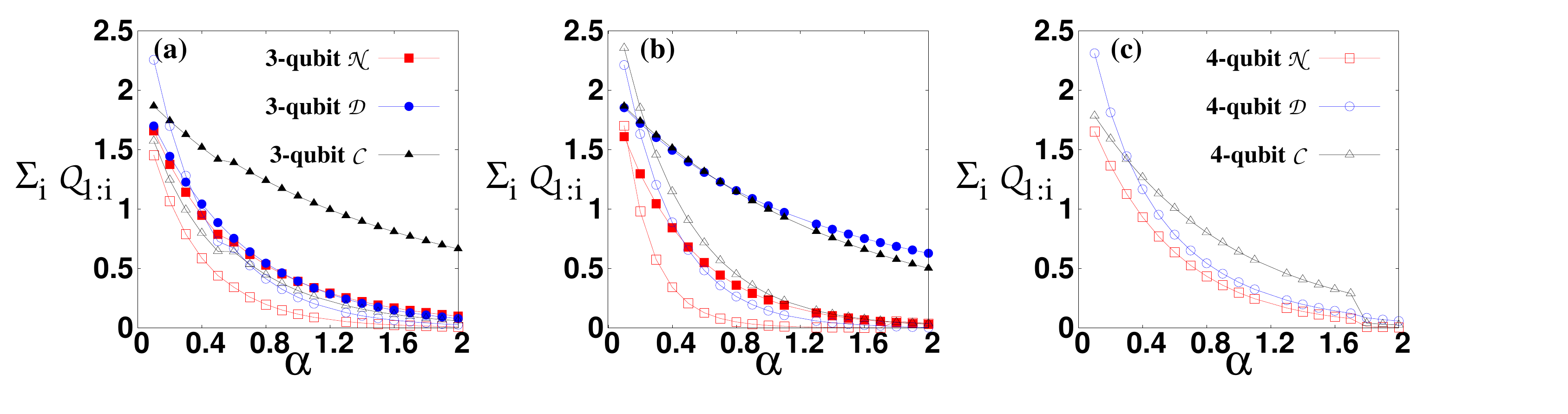}
\caption{\textcolor{black} {Plot of $\sum_i \mathcal{Q}_{1:i}$ (ordinate) for three-qubit (solid points) and four-qubit (hollow points) states
against $\alpha$  (abscissa) in case of negativity $\mathcal{N}$ (squares), discord $\mathcal{D}$ (circles), and concurrence $\mathcal{C}$ (triangles). The states considered are (a) Haar uniformly generated random states, (b) Dicke states with single excitation, and (c) Dicke-states with two excitations.  The vertical axis is in ebits in case of entanglement (negativity and concurrence), and in bits in case of discord while the horizontal axis is dimensionless.}  }
\label{monosum}
\end{figure*}

 \subsection{Maximum of the sum of bipartite QC measures}

To understand  the criticalities deeply, let us consider the actual maximum of \(\sum_{i=2}^{N} \mathcal{Q}_{1i}\) and its difference with the algebraic maximum. It is clear from previous investigations that the  monogamy-based bound is too stringent when \(\alpha\) is small. However, the sum of bipartite QCs  is still  lower than the sum of the individual maxima, i.e.  \(\sum_{i=2}^{N} \mathcal{Q}_{1i} < (N-1)\) for any QC measure in a qubit-scenario.  Since the monogamy score is negative for low values of $\alpha$, the above quantity has a substantial strength at those values and decreases sharply with a rise in the exponent. With higher number of parties in the parent multipartite state, it decreases for moderate to high $\alpha$ and also drops down to zero much more rapidly. For negativity, it is lower than that of concurrence. Our observations are illustrated in Fig. \ref{monosum} for random states, and Dicke states.


\subsection{Significance of the monogamy exponent $\alpha$}
\label{subsec:sig_alpha}

\textcolor{black}{Monogamy is a quantum property which can be demonstrated by using certain quantum correlations as was initially proposed by Coffman, Kundu, and Wootters \cite{CKW}. At the same time, $\delta_{\mathcal{Q}^\alpha} = \mathcal{Q}^{\alpha}_{1:\mbox{rest}} - \sum_{i=2}^N \mathcal{Q}_{1:i}^\alpha$, which we refer to as the monogamy score of $\mathcal{Q}$, can be used to distinguish classes of states. For example, $\delta_{\mathcal{Q}^\alpha}$ with squared concurrence (i.e., $\alpha = 2$ and $\mathcal{Q} \equiv~\text{concurrence}$) vanishes for all states from the W-class while it is non-vanishing for the GHZ-class states, thereby differentiating between two SLOCC inequivalent classes \cite{DurVidalCirac}. In this case, both classes satisfy the monogamy inequality. \\
In the case of quantum discord \cite{ZurekDiscord}, W-class states always violate the inequality, and are thus referred to as non-monogamous states in literature \cite{discrevus}. Moreover, we find that the monogamy score shifts from being negative to positive with an increase in the number of parties, independent of the correlation measure chosen. This may be attributed to an increase in inherent quantum correlations present in the states. Only states with low value of genuine multipartite entanglement (GME), possess a negative monogamy score when a large number of parties are involved. This suggests that we can gather an idea about the amount of correlation present in a state by examining its monogamy score.\\
The monogamy inequality considered with the square of the concurrence i.e., $\alpha = 2$. Thus, it is important to extend the idea to other powers of different quantum correlation measures and study the behavior of states in terms of the monogamy inequality. It was shown  \cite{salini} that, for a given QC measure, there always exists an exponent for which all states satisfy the monogamy inequality. The work by Rethinasamy et. al. \cite{unimonoSappy} found two such values of $\alpha$, which provided  bounds on the exponent characterising monogamous and non-monogamous states. This establishes the fact that the behavior of states as well as correlations depends heavily on the exponent considered in the monogamy inequality.\\
In our work, we try to explore the performance of states in terms of the monogamy criterion, for $\alpha$ values different from the bounds defined in the aforementioned article. We observe a dependence of the monogamy inequality on the exponent as well as on the number of parties comprising the states. Thus for a given class of states and a fixed quantum correlation measure, the monogamy inequality also becomes a function of the exponent. For example, we observe that for Haar uniformly generated states, the monogamy inequality is satisfied by almost all the states, whereas for states belonging to the W-class, a much higher value of $\alpha$ is required to satisfy the same inequality. This feature shows that monogamy can be well characterised by the exponent in the inequality. \\}

	\section{Interplay between Measurement-,  geometry- and monogamy-based quantum correlations}
	\label{sec_LE}
	
	Let us now move to relate  measure-based QCs  with both the monogamy-based QC measures and geometric measure of entanglement. The measurement-based measures as well as geometric measures quantify QCs in an active way while monogamy-based measures do the job in a passive way as explained in the introduction. This is due to the fact that
	instead of tracing out $N-2$ parties and looking at algebraic combinations of bipartite QCs, we now shift our attention to quantum correlations which are obtained by employing optimal local  projective measurements on the $N-2$ qubits of the $N$-qubit state. 
These local measurements concentrate the global correlations of the state into a particular bipartite pair and are known as localizable correlations \cite{entass, LC1, LC2}. Therefore, the localized bipartite correlations  have potential to capture quantumness distributed in  multipartite states  \cite{EoAGour,  Ratul}.

Since we want to relate measurement-based QC measures with the monogamy-based one,  we introduce a localized version of QC measure, $\mathcal{Q}$, with a power $\alpha$, denoted by $\mathcal{L}\mathcal{Q}^\alpha$, when the local measurements are performed in the all the parties except first two parties $1$ and $2$, and for a multipartite pure state, \(|\psi_N\rangle\) and given QC measure, \(\mathcal{Q}\), it can mathematically be represented as 
\begin{eqnarray}
\mathcal{L}\mathcal{Q}^\alpha(|\psi_N\rangle) = \max_{\{\Pi\}} \sum_{k = 1}^{2^{N-2}} p_k \mathcal{Q}^{\alpha}(|\phi_k\rangle), 
\end{eqnarray}
where $\{\Pi\}$ denotes the set of local rank-1 projective measurements on the $N-2$ qubits,  the binary equivalent of $k$ is a particular outcome combination of the $N-2$ qubit projectors,  and $|\phi_k\rangle$ is the normalized post measurement state for the $k^{\text{th}}$ outcome with $p_k$ being the corresponding probability.
We report the connection of $\mathcal{L}\mathcal{Q}^\alpha$ with \(\mathcal{G}\) and  $\delta_{\mathcal{Q}^{\alpha}}$, as well as the variation of $\mathcal{L}\mathcal{Q}^\alpha$  with the power, $\alpha$. For concurrence, negativity and discord  as QC measures, the respective localized versions are  denoted by \(\mathcal{L}\mathcal{C}^\alpha(|\psi_N\rangle)\), \(\mathcal{L}\mathcal{N}^\alpha(|\psi_N\rangle)\), and \(\mathcal{L}\mathcal{D}^\alpha(|\psi_N\rangle)\). \\
\textcolor{black}{Localisable entanglement and GGM seem to be two unrelated quantities by definition. For example, if we consider $|\psi\rangle = |\phi^+\rangle \otimes |0\rangle$, where $|\phi^+\rangle = (1/\sqrt{2})(|00\rangle + |11\rangle)$, then a $\sigma_z$ measurement on the third party and subsequent tracing out, we obtain a maximally entangled state, even though the GGM of $|\psi\rangle$ is zero. On the other hand, a $\sigma_x$ measurement on the third party of the three-qubit GHZ state, and subsequent tracing out also furnishes a maximally entangled state and the GGM of the GHZ state is also maximal having a value of $0.5$. Therefore, in case of both vanishing GGM as well as very high GGM, it is possible to localise high amount of entanglement. Thus,  there seems to be no relation between the two quantities if we consider extreme situations. The main idea of this section is to show that  the situation is not so despairing if one considers random states. }

\begin{figure}[h]
		\centering
		\includegraphics[width=\linewidth]{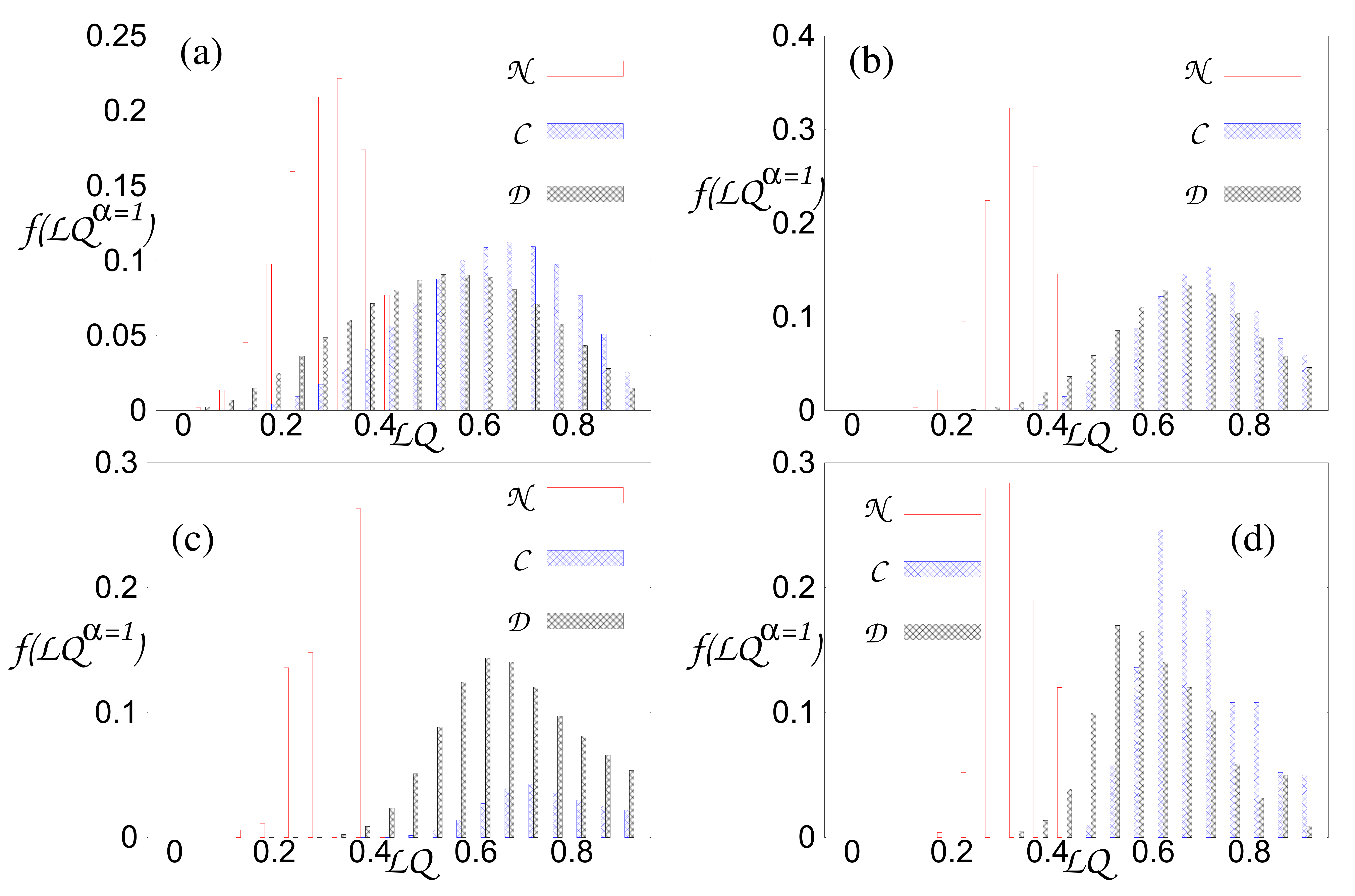}
		\caption{\textcolor{black}{ \(f(\mathcal{LQ}^{\alpha=1})\)(ordinate)  vs. localizable correlation $\mathcal{LC}$. The QCs localized here are negativity $\mathcal{N}$ (blank),  concurrence $\mathcal{C}$ (light), and  discord $\mathcal{D}$ (dark) for random states of three-qubits (a), four-qubits (b), five-qubits (c) and six-qubits (d).
All other specifications are the same as in Fig. \ref{fig_ggm_3456}. The abscissa is in ebits (for negativity and concurrence) or bits (for discord) whereas the ordinate is dimensionless.} 		
		}
		\label{lggm_3456}
	\end{figure}

%
	
	\subsection{Relation of $\mathcal{L}\mathcal{Q}^\alpha$  with $\delta_{Q^\alpha}$ and GGM}
	

\emph{Frequency distribution of localizable QCs. } Before performing this relational analysis in a systematic way, let us study the frequency distribution of $\mathcal{L}\mathcal{Q}^{\alpha =1}$  (see \cite{Ratul} for entanglement of formation). We find that  like monogamy score and GGM, the shape of the distribution for random states is bell-like and it shifts towards its algebraic maximum with the increase of $N$  and becomes  sharper with $N$ since the average value of   $\mathcal{L}\mathcal{Q}^{\alpha =1}$ increases and SD decreases with the increase of number of parties as shown in Fig. \ref{lggm_3456}. The observation is independent of the choice of QC measures and for different values of \(\alpha\) for Haar uniformly generated random states. The opposite picture emerges for  the Dicke states with low excitations ---  $\langle \mathcal{L}\mathcal{Q}^{\alpha =1} \rangle$ decreases with $N$, i.e., the distribution shifts towards the low value of the respective measure for high $N$ although the width of the distribution decreases with the increase of number of parties. However, the increase of mean with $N$ is much slower than the one observed for GGM. For example, the average obtained for negativity and discord  for Haar uniformly generated states are respectively \(0.337, 0.378, 0.397\) and \(0.58, 0.714, 0.727\) with \(N=3, 4, 5\) (compare them with Table \ref{msdTGHZ}).  With the increase of \(\alpha\), mean decreases and SD increases both for random and Dicke states. 

\emph{ Relation of measurement-based QCs with generalized geometric measure as well as monogamy score.} In stark contrast to the relation of monogamy score and GGM, measurement-based QCs  behave differently with GGM and monogamy score. Specifically, in \((\mathcal{G} (\delta_{\mathcal{N}^1}), \mathcal{L}\mathcal{Q}^{\alpha =1})\)-plane, random states are scattered, thereby showing that states with low GGM (monogamy score) can result with high amount of localizable entanglement and at the same time, states with high multipartite entanglement are able to localize small amount of entanglement  as depicted in Figs. \ref{fig_c3456} and \ref{fig_conc_ggm_score}, irrespective of number of parties. Such a picture only changes when we consider the Dicke state with a single excitation which only displays a triangular structure, thereby showing a forbidden region in that plane. It implies that although states having low GGM can concentrate high localizable entanglement, a state with high GGM can always produce moderate amount of entanglement for Haar uniformly generated Dicke states. 


Next we will   argue that the localizable entanglement (measured either by concurrence or negativity) can have substantial value for sufficiently small GGM in case of random three-qubit states as well as three-qubit $|\psi_D^1\rangle$. \\

	\begin{proposition}
For arbitrary three-qubit pure states, localizable entanglement can have a moderately high value even when the genuine multipartite entanglement content of the state is  small. 
\end{proposition}

\begin{proof}
The Schmidt decomposition for a tripartite pure state is given by \cite{acinschmidt}
	\begin{eqnarray}
		| \psi_3 \rangle &= & a_1 |000\rangle + a_2 \exp^{i \phi} |100\rangle + a_3|101 \rangle + \nonumber \\
		&& a_4|110\rangle + a_5|111\rangle
\label{3schmidt}
	\end{eqnarray}
	where all parameters are real and positive semidefinite with $0 \leq \phi \leq 2\pi$ and $\sum_i a_i^2 = 1$.
By performing projective measurement on the third qubit of $|\psi_3 \rangle$, the localizable concurrence of the remaining two qubits is given by $2\sqrt{\det(\rho_1)}$ where $\rho_1 = Tr_{2,3}|\psi \rangle_M \langle \psi|$, with $|\psi\rangle_M$ denoting the post-measurement state.
Suppose that  \(\mathcal{LC}\) achieves its optimum value due to measurements along the $X$, $Y$ or $Z$ direction, i.e. in the eigenvectors of \(\sigma_i,\, i=x, y, z\). Incidentally, for all three cases, the localizable concurrence is given by
\begin{equation}
\mathcal{LC}_{\sigma} = 2a_1a_4.
\label{rhoLC}
\end{equation}
The actual \(\mathcal{LC}\) can be higher than \(\mathcal{LC}_{\sigma}\), i.e.,
 $\mathcal{LC}_{\sigma} \leq \mathcal{LC}$.\\
 
 \begin{figure}[h]
			\centering
			\includegraphics[width=\linewidth]{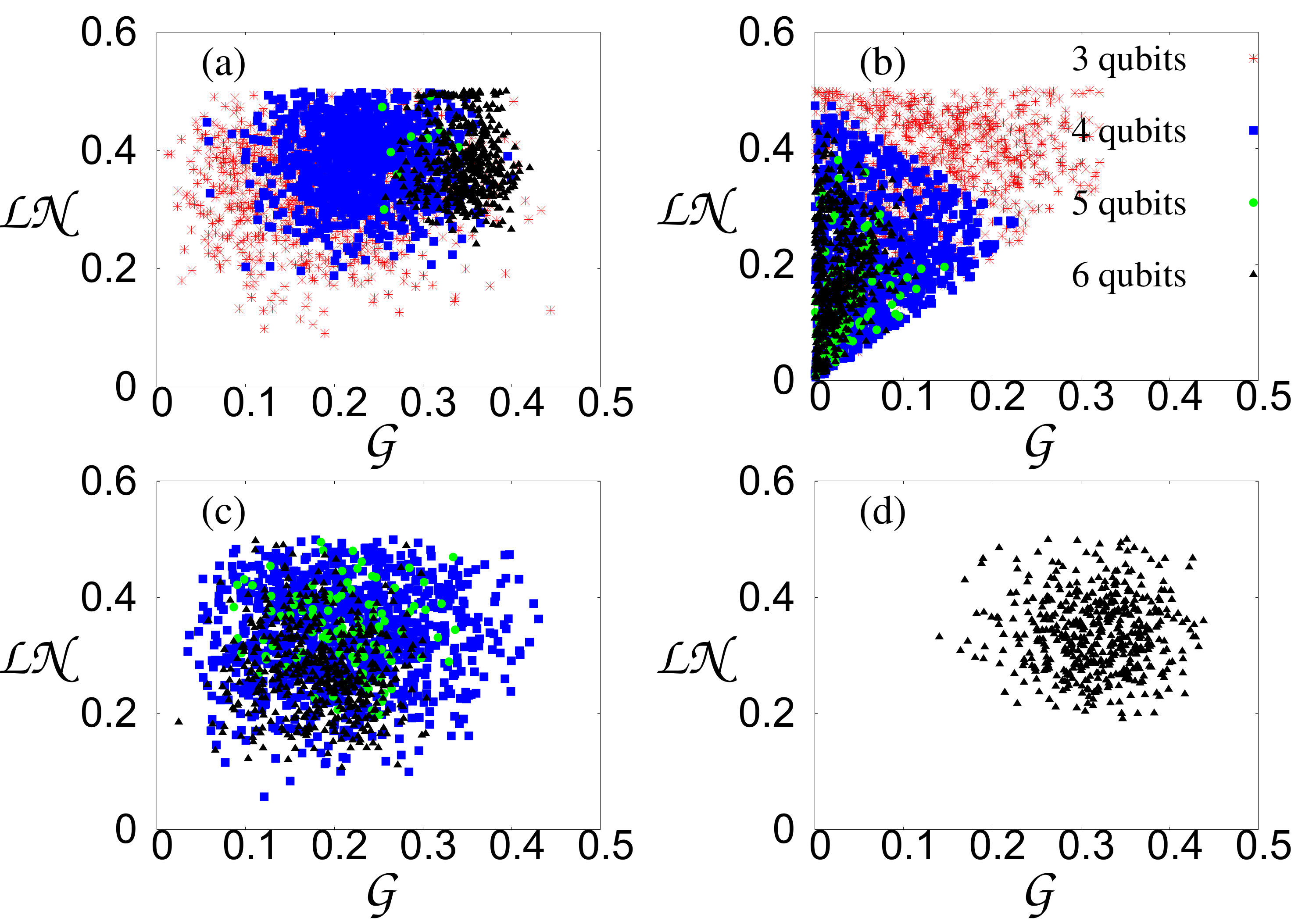}
			\caption{\textcolor{black}{Scattered plot of  localizable negativity, \(\mathcal{LN}\), (\(y\)-axis) against GGM $\mathcal{G}$ (\(x\)-axis) for (a) random states , (b) random  Dicke states with a single excitation, (c) Dicke states with  two excitations, and (b) Dicke states with three excitations for three-qubits (stars), four-qubits (squares), five-qubits (circles) and six-qubits (triangles).  Both the axes are in ebits.} }
			\label{fig_c3456}
		\end{figure}

To obtain $\mathcal{G}(| \psi \rangle)$, we note the eigenvalues of the single qubit reduced density matrices corresponding to the state in Eq. \eqref{3schmidt} to be
\begin{eqnarray}
&& \lambda_1^\pm = \frac{1}{2}(1 \pm \sqrt{1 - 4\mathcal{LC}_{\sigma} - f_1(a_i)}), \\
&& \lambda_2^\pm = \frac{1}{2}(1 \pm \sqrt{1 - \mathcal{LC}^2_{\sigma} - f_2(a_i)}), \\
&& \lambda_3^\pm = \frac{1}{2}(1 \pm \sqrt{1 - \mathcal{LC}^2_{\sigma} - f_3(a_i)}),
\end{eqnarray}
where we have clubbed all the terms which cannot be written in terms of $\mathcal{LC}$ into $f_i(a_i)$. Hence, 
the GGM is $\mathcal{G}(|\psi\rangle) = 1 - \lambda_i^+ = \lambda_i^-$. Depending on the values of the coefficients, any one $\lambda_i^+$ can give be maximum and that contributes to the GGM. By ignoring $f_i(a_i)$ which are typically a very small numbers,  the relationship between the modified GGM and localizable concurrence is found to be
\begin{eqnarray}
&& \mathcal{LC}_{\sigma} = \sqrt{1 - (1 - 2\mathcal{G}_1(|\psi_3\rangle))^2} \; \; \; \; \text{if} \;\lambda_{2,3}^+ \;\;\text{is maximum}, \nonumber \\
&& \mathcal{LC}_{\sigma} = \frac{(1 - (1 - 2\mathcal{G}_1(|\psi_3\rangle))^2)}{2} \; \; \; \; \text{if} \;\lambda_{1}^+ \;\;\text{is maximum} \nonumber, \\
\end{eqnarray}
where \(\mathcal{G}_1 (|\psi_3\rangle) \geq \mathcal{G}(|\psi_3\rangle)\). Analysing the above relations geometrically, we observe, that even for values of $\mathcal{G}_1(|\psi_3\rangle) \leq 0.1$, the \(\mathcal{LC}\) with restricted set of  measurement can be \(0.6\) or even higher. Since the original \(\mathcal{LC}\) can be higher while the GGM can also take a lower value than the actual one, the above argument shows that 
sizeable correlations can be localized even if the original state possesses insignificant multipartite entanglement. 

\end{proof}

\begin{figure}[h]
			\centering
			\includegraphics[width=\linewidth]{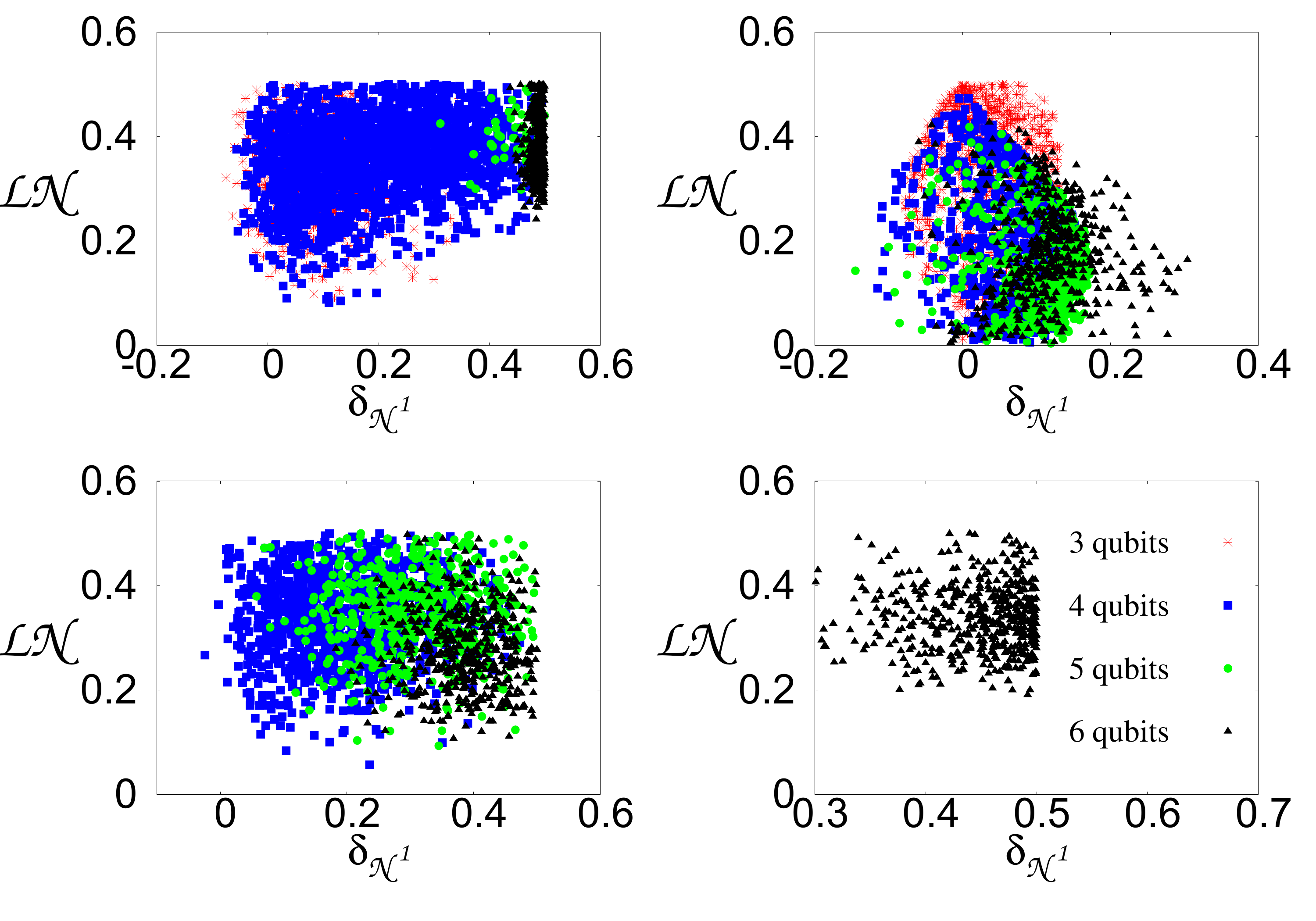}
		\caption{\(\mathcal{LN}\), (vertical) vs. $\delta_{\mathcal{N}^1}$ (horizontal). All other specifications are similar to Fig. \ref{fig_c3456}. }  
			\label{fig_conc_ggm_score}
		\end{figure}

\emph{Remark.} The similar argument can be given to have the relation between localizable negativity and GGM. Since we have already established a relation between GGM and monogamy score, the above results also imply that it is possible to find states having low monogamy score which can produce corresponding high localizable quantum correlations  (see Fig. \ref{fig_conc_ggm_score}).


Like arbitrary three-qubit states, the three-qubit Dicke state with a single excitation having low genuine multipartite entanglement can produce high localizable entanglement as shown in Fig. \ref{fig_c3456}. To show that, let us consider the three-qubit Dicke state, 
\(|\psi_D^1\rangle = a_1 |100 \rangle + a_2 |010 \rangle + a_3 |001 \rangle\) with $\sum_i a_i = 1$. 
In this case, by assuming $a_1,a_2 \geq a_3$, we have
\begin{eqnarray}
&& \mathcal{G}(|\psi_D^1\rangle) = a_3^2 = 1 - a_1^2 - a_2^2 \\
&& \mathcal{LC}_{\sigma}(|\psi_D^1\rangle) = 2a_1 a_2
\end{eqnarray}
Some  algebra then allows us to end up with the relation between $\mathcal{G}$ and the localizable concurrence as
\(\mathcal{LC}_{\sigma} = \mathcal{G}(|\psi_D^1\rangle) + (2a_1 a_2) - a_3^2\).  
The above relation shows the linear dependence of $\mathcal{LC}_{\sigma}$ on $\mathcal{G}$ as depicted in Fig. \ref{fig_c3456}. Since  $a_1,a_2 \geq a_3$, we have $(2a_1 a_2) - a_3^2 \geq 0$ and thus, the dependence of \(\mathcal{LC}_{\sigma} \leq \mathcal{LC} \) on \(\mathcal{G}\) also shows that the localizable concurrence easily exceeds the GGM.
	
	\begin{figure}[h]
		\centering
		\includegraphics[width=\linewidth]{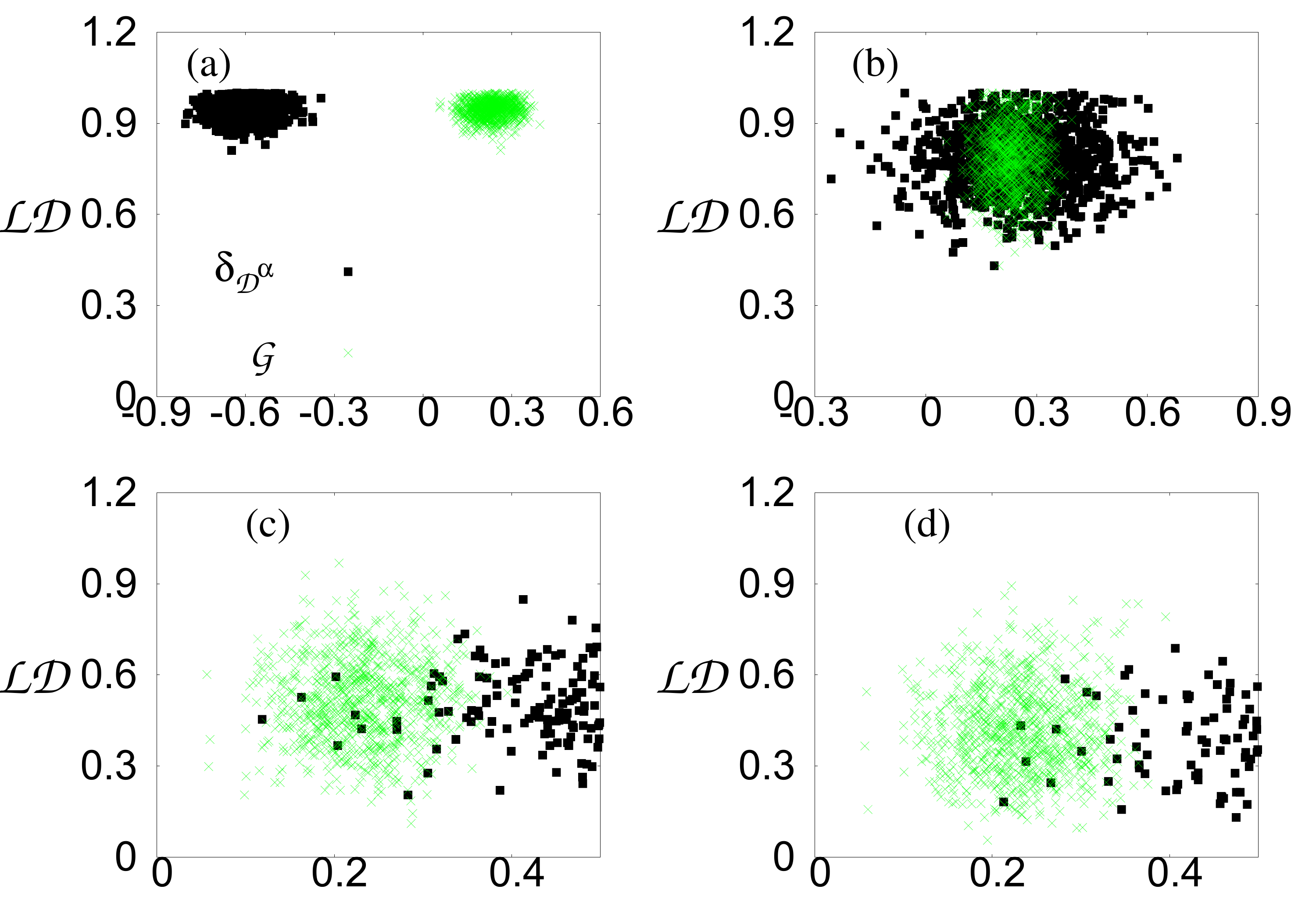}
		\caption{\textcolor{black}{Scattered diagram of localizable discord $\mathcal{LD}$ (ordinate) against $\delta_{\mathcal{D}^\alpha}$ (dark) and $\mathcal{G}$ (light) along abscissa. The choices of  $\alpha$ are (a) 0.1, (b)  0.5, (c) 1.1, and (d) 1.5 where four-qubit  Haar uniformly simulated random states are considered.}}
		\label{3conc_sg}
	\end{figure}	
	
	\begin{figure}[h]
		\centering
		\includegraphics[width=\linewidth]{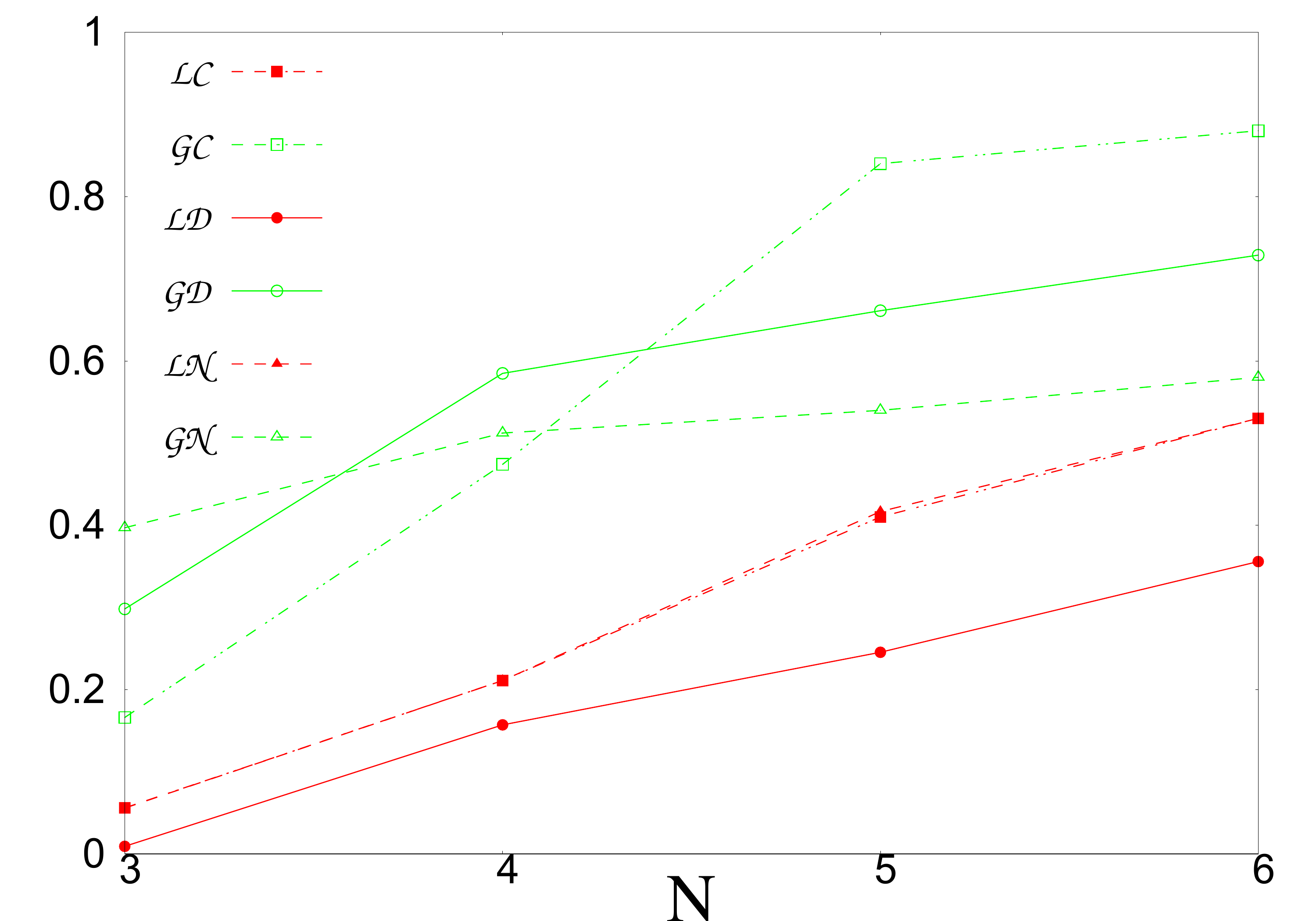}
		\caption{\textcolor{black}{For a fixed $N$ (abscissa), we study the minimum localizable QC (light) that can be obtained and its corresponding GGM (dark) along the ordinate. \(\mathcal{G}x,\, \, x= \mathcal{C}, \mathcal{N}, \mathcal{D}\) denote the GGMs when minimum localizable concurrence ($\mathcal{LC}$), localizable negativity ($\mathcal{LN}$) and localizable discord ($\mathcal{LD}$) are achieved.    and  The \(y\)-axis is in ebits and the \(x\)-axis is dimensionless.} }
		\label{min_c_ggm}
	\end{figure}
	

As discussed qualitatively and also in Proposition 2, the connection between GGM (monogamy score) and \(\mathcal{LQ}\) does not have any definite structure. To make their comparison more quantitative, we consider two situations --- 1. For a fixed $N$, we  find  minimum and maximum localizable QC that can be achieved and the corresponding genuine multipartite entanglement content of a state; 2. for a given range of GGM values, minimum and maximum QC that can be localized are focussed on. In particular,  we report that the GGM at which the minimum of the correlation occurs increases with number of parties for random states, as shown in Fig. \ref{min_c_ggm}. It is due to the fact that among random states, average GGM also increases with $N$.  As  seen from  Table \ref{min_lc_ggm},  to localize nonvanishing QC, a very small amount of GGM is required. Moreover, we observe that to localize minimum amount of QC in the first and the second qubits, the GGM required is always higher than the amount of LQC, i.e., \(\mathcal{LQ}^{\min} < \mathcal{G}\).  
Secondly, if we fix GGM in a certain range, the minimum localizable QC can also follow the similar trend, i.e. to localize QC minimally, the corresponding GGM required for that is substantial. 
On the other hand, for a fixed GGM value, the localizable QCs can always reach their corresponding maximum value. 

\emph{Effects of exponents on QCs in localizable quantity. } We now investigate the effect of varying $\alpha$ introduced in the localizable QCs and we consider the same \(\alpha\) in monogamy score. The trend that we observed in Proposition for arbitrary states or Dicke state remains same by varying \(\alpha\). In particular, 
when we have low \(\alpha <1\), highly localizable entangled states are generated for varying GGM and monogamy score while for high \(\alpha\), low amount of entanglement can be localized. Such observation is possibly artifact  of  the functional form of the QC measures as also the case of monogamy scores (see Fig. \ref{3conc_sg}). 



%
 Similar to Fig. \ref{alpha_C},  we now ask a question: for a given localizable QC, what is critical exponent above which, the monogamy score is always nonnegative? We find that unlike GGM, no such universal picture emerges as we hinted by the relation between localizable QC and GGM. In this case, states with high localizable QC can require high \(\alpha\) to make the states monogamous (see Fig. \ref{fig_LQCvsalpha_C}). Only low \(\alpha\) is required for high $N$ in case of $\delta_{\mathcal{Q}}$, which one expects from the behavior of monogamy score itself. This observation possibly indicates that localizable QC measure has some component which are due to multipartite state but is qualitatively different than multipartite entanglement monotones (cf. \cite{EoAGour}). Notice also that such conclusion may be changed if we alter the definition of localizable QC (cf. \cite{popescu} and references thereto).
 
 \begin{figure}[h]
 	\centering
 	\includegraphics[width=\linewidth]{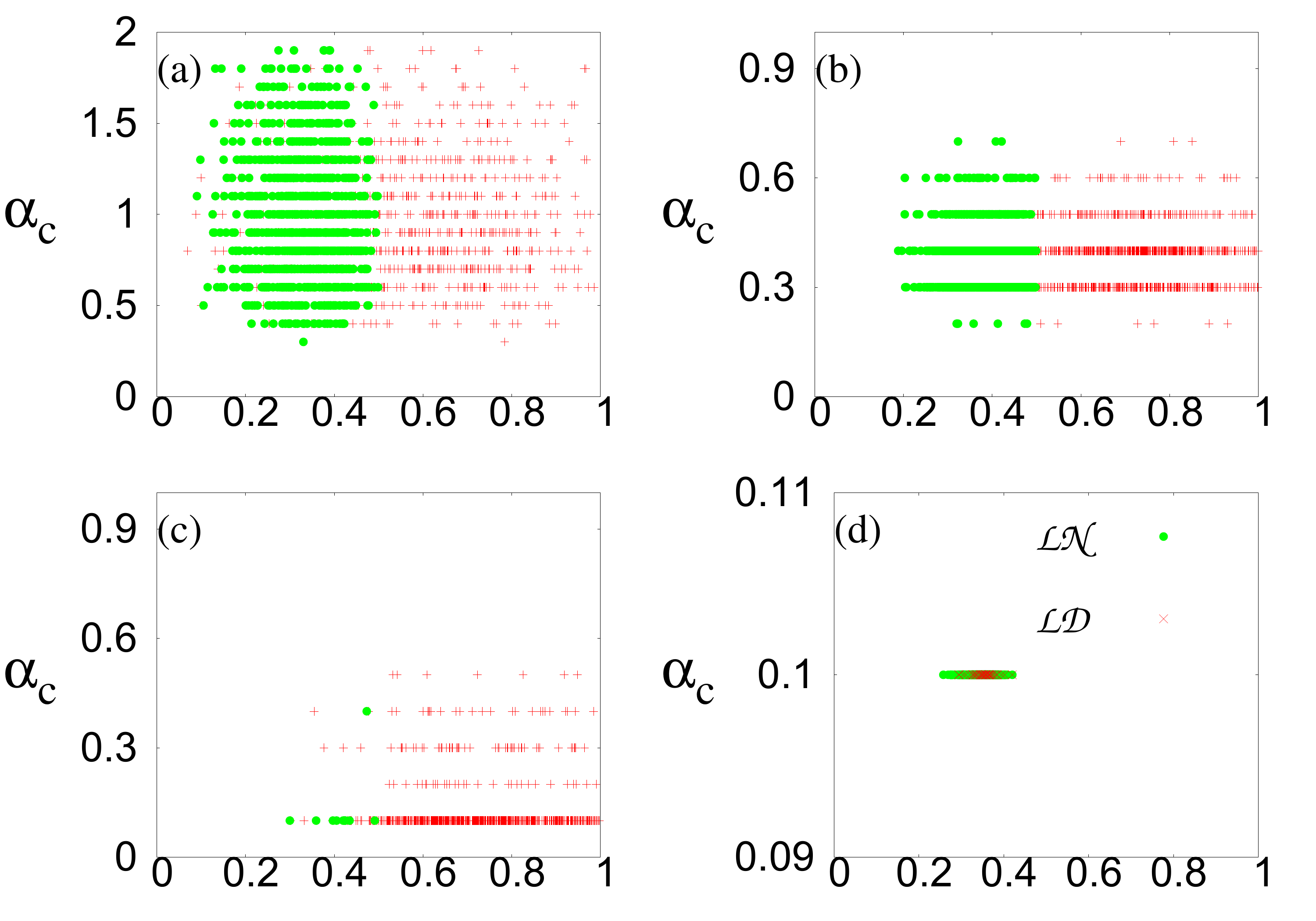}
 	\caption{\textcolor{black}{Critical exponent $\alpha_C$ (ordinate) is plotted against $\mathcal{LN}$ (light) and $\mathcal{LD}$ (dark) for random states of (a) three-qubits, (b) four-qubits, (c) five-qubits, and (d) six-qubits. All other coordinates are the same as in Fig. \ref{fig_ggm_3456}. The x-axis is in ebits for $\mathcal{LN}$ and in bits for $\mathcal{LD}$ whereas the y-axis is dimensionless.}   }
 	\label{fig_LQCvsalpha_C}
 \end{figure}


	\begin{table}[]
		\caption{$\mathcal{G}$ for the corresponding  minimum localizable QCs.} 
		\begin{tabular}{|l|l|l|l|}
			\hline
			N & \multicolumn{1}{c|}{$\mathcal{LC}$} & \multicolumn{1}{c|}{$\mathcal{L}N$} & \multicolumn{1}{c|}{$\mathcal{LD}$} \\ \hline
			&                                     &                                     &                                     \\ \hline
			3 & 0.083                               & 0.198441                            & 0.149041                            \\ \hline
			4 & 0.237                               & 0.256076                            & 0.292371                            \\ \hline
			5 & 0.33                                & 0.259891                            & 0.330505                            \\ \hline
			6 & 0.36                                & 0.27                                & 0.328931                            \\ \hline
		\end{tabular}
	\label{min_lc_ggm}
	\end{table}
	


	\subsection{Deviation from algebraic maximum}

Let us now  investigate the behavior of $\sum_{i=2}^N \mathcal{LQ}_{1i}$ for an N-partite state. The algebraic maximum of this quantity is $(N-1)$ which can be achieved by the GHZ state. However, the actual bound turns out to be quite different for Haar uniformly generated states. We observe that the sum falls short of its algebraic value for all classes of states, especially for high $N$ (see Table \ref{max_le_sum_table}) .

For random three-qubit states, the actual upper bound  is close to its algebraic maximum, i.e.,  $2$ although the difference between algebraic maxima and the maximum obtained numerically  increases with the number of parties. E.g., if we consider concurrence as a quantum correlation measure, the gap is 0.007 for three-qubit random states while it rises to 0.48 in case of six-qubits (see Fig. \ref{le_sum_hist}). In case of Dicke state, the gap turns out to be significant, i.e., it fails to attain the algebraic threshold by a big margin. In this instance too, the difference increases with the number of constituent qubits. E.g., the sum reaches only about half of the algebraic maximum for $| \psi_{D}^1\rangle$. However, with more number of excitations in Dicke states, picture similar to random states develops. 

\begin{table*}[]
	\caption{Maximum of $\sum_{i=2}^N \mathcal{LQ}_{1i}$.} 
	\begin{tabular}{|l|r|r|r|r|r|r|r|r|r|r|r|r|}
		\hline
		& \multicolumn{4}{c|}{$\sum_{i=2}^N \mathcal{LN}_{1i}$}                                                                                               & \multicolumn{4}{c|}{$\sum_{i=2}^N \mathcal{LC}_{1i}$}                                                                                               & \multicolumn{4}{c|}{$\sum_{i=2}^N \mathcal{LD}_{1i}$}                                                                                               \\ \hline
		& \multicolumn{1}{c|}{}       & \multicolumn{1}{c|}{}           & \multicolumn{1}{c|}{}           & \multicolumn{1}{c|}{}           & \multicolumn{1}{c|}{}       & \multicolumn{1}{c|}{}           & \multicolumn{1}{c|}{}           & \multicolumn{1}{c|}{}           & \multicolumn{1}{c|}{}       & \multicolumn{1}{c|}{}           & \multicolumn{1}{c|}{}           & \multicolumn{1}{c|}{}           \\ \hline
		& \multicolumn{1}{c|}{Random} & \multicolumn{1}{c|}{$\psi_D^1$} & \multicolumn{1}{c|}{$\psi_D^2$} & \multicolumn{1}{c|}{$\psi_D^3$} & \multicolumn{1}{c|}{Random} & \multicolumn{1}{c|}{$\psi_D^1$} & \multicolumn{1}{c|}{$\psi_D^2$} & \multicolumn{1}{c|}{$\psi_D^3$} & \multicolumn{1}{c|}{Random} & \multicolumn{1}{c|}{$\psi_D^1$} & \multicolumn{1}{c|}{$\psi_D^2$} & \multicolumn{1}{c|}{$\psi_D^3$} \\ \hline
		& \multicolumn{1}{l|}{}       & \multicolumn{1}{l|}{}           & \multicolumn{1}{l|}{}           & \multicolumn{1}{l|}{}           & \multicolumn{1}{l|}{}       & \multicolumn{1}{l|}{}           & \multicolumn{1}{l|}{}           & \multicolumn{1}{l|}{}           & \multicolumn{1}{l|}{}       & \multicolumn{1}{l|}{}           & \multicolumn{1}{l|}{}           & \multicolumn{1}{l|}{}           \\ \hline
		\multicolumn{1}{|r|}{3} & 0.997                       & 0.707                           &                                 &                                 & 1.993                       & 1.414                           &                                 &                                 & 1.995                       & 1.563                           &                                 &                                 \\ \hline
		\multicolumn{1}{|r|}{4} & 1.481                       & 0.866                           & 1.47                            &                                 & 2.973                       & 1.732                           & 2.94                            &                                 & 2.946                       & 1.654                           & 2.22                            &                                 \\ \hline
		\multicolumn{1}{|r|}{5} & 1.944                       & 1.1                             & 1.917                           &                                 & 3.94                        & 2                               & 3.83                            &                                 & 3.845                       & 1.939                           & 2.465                           &                                 \\ \hline
		\multicolumn{1}{|r|}{6} & 2.34                        & 1.398                           & 2.104                           & 2.22                            & 4.52                        & 2.58                            & 4.17                            & 4.3                             & 4.14                        & 2.435                           & 2.94                            & 4.247                           \\ \hline
	\end{tabular}
	\label{max_le_sum_table}
\end{table*}

\begin{figure}[h]
	\centering
	\includegraphics[width=\linewidth]{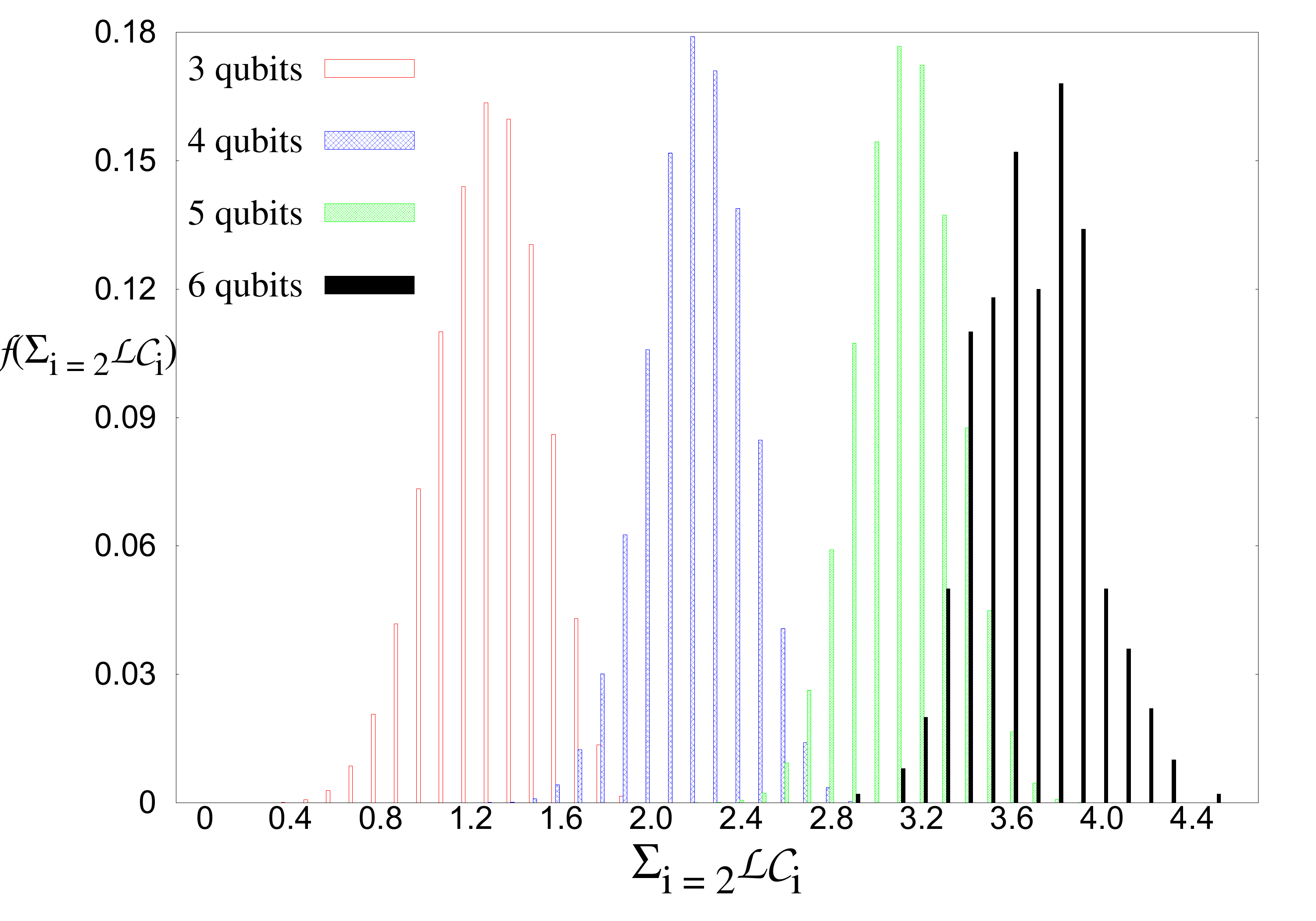}
	\caption{ $f(\sum_{i=2} \mathcal{LC}_i)$  (ordinate) against $\sum_{i=2} (\mathcal{LC})_i$ (abscissa) for three-(red), four- (blue), five- (green) and six- qubit (black) random states. It also shows that although for three-qubits, sum is close to the algebraic maximum, the gap between algebraic maxima and the maximum of the sum obtained via random states increases with the increase of the number of parties. The ordinate is dimensionless while the abscissa is in ebits.  }
	\label{le_sum_hist}
\end{figure}

	\section{Conclusion}
	\label{sec_conclu}
    
     In a multipartite domain,   quantum correlations (QC) even for  pure states cannot be characterized in a unique way. Over the years, several quantifications from different origins  have been  proposed which elucidate specific features of quantum states, important for building  quantum technologies. In this work, we provide a connection between three such independent quantum correlation measures, defined from different perspectives, thereby bringing them under a single umbrella. In particular, we choose a geometry-based entanglement measure quantifying genuine multipartite entanglement, monogamy-based quantum correlation measures with different exponents and measurement-based  measures. Both monogamy- and measurement-based measures are constructed by considering both entanglement and other quantum correlation measures. 
     
     We reported that there exists a critical content of genuine multipartite entanglement above which no multipartite states violate monogamy inequality. Typically, monogamy relations for a quantum correlation  are considered with an exponent that can be thought of as "ad-hoc" \cite{MonoFaithful}.  We find that for a fixed genuine multipartite entanglement, there always exists a critical exponent above which all measures satisfy monogamy relation. For Haar uniformly generated states, such a critical exponent decreases with the increase of number of parties. We also proved that if an arbitrary three-qubit state and generalized Greenberger-Horne-Zeilinger states (gGHZ) possess the same amount of genuine multipartite entanglement, then the entanglement monogamy score of the former is bounded above by that of the latter. Such a hierarchy between random states and gGHZ states does not hold for states with a higher number of parties. The back of the envelope calculations  also reveals that average  genuine multipartite entanglement content of random states with arbitrary number of parties coincides with the Dicke states having half of the sites excited. 
     
     On the other hand, we showed that a state having low multipartite entanglement content can localize a high amount of quantum correlations in two parties and vice-versa. This result indicates that localizable quantum correlations can have some components carrying multipartite characteristics of states although it also highlights the difference between genuine multipartite entanglement and localizable entanglement. Notice that a different process of sweeping entanglement towards two parties than the one considered in this work may show different characteristics.   Interestingly, we observe that the monogamy score of QCs behaves more like multipartite  measures than localizable QCs. Specifically, we observed that states having high localizable entanglement may require a high critical exponent to satisfy the corresponding monogamy inequality, thereby showing its different nature from genuine multipartite entanglement. Moreover, we found that the sum of bipartite localizable entanglement of multipartite random states are bounded above by a quantity which is close to its algebraic maximum and the gap between algebraic and the actual bounds increases with the number of parties for random states while the difference is substantial for Dicke states. 

 Since all the quantum correlation quantifiers have different kinds of importance in quantum information science, the connection established in this paper possibly 
 gives a hint towards choosing the QC measure, depending on their tasks, instead of their amount.

	\section{Acknowledgement}
	
	We acknowledge the support from Interdisciplinary Cyber Physical Systems (ICPS) program of the Department of Science and Technology (DST), India, Grant No.: DST/ICPS/QuST/Theme- 1/2019/23 and SM acknowledges Ministry of Science and Technology in Taiwan (Grant no. 110-2811-M-006 -501). We  acknowledge the use of \href{https://github.com/titaschanda/QIClib}{QIClib} -- a modern C++ library for general purpose quantum information processing and quantum computing (\url{https://titaschanda.github.io/QIClib}) and cluster computing facility at Harish-Chandra Research Institute. 
	
\bibliography{bib1}
\bibliographystyle{apsrev4-1}

\end{document}